\documentclass[12pt,a4paper]{article}
\usepackage[english]{babel}

\usepackage{pstricks}
\usepackage{pst-node}
\usepackage{amsmath}
\usepackage{amssymb}
\usepackage{graphicx}
\usepackage{enumerate}
\usepackage[utf8]{inputenc}
\usepackage{empheq}
\usepackage{algorithm,algorithmic}
\usepackage{epstopdf}

\usepackage{caption}
\usepackage[text={15cm,24cm}]{geometry}

\newenvironment{proof}{\noindent{\bf Proof}:\ }%
   {~\ \hfill $\Box$\vspace{0,5cm}}

   {~\ \hfill {\em End~of~Claim}}

    {~\ \hfill$\Diamond$\vspace{0,5cm}}

\def\LStatex{\Statex\unskip\the\therules}

\newtheorem{prop}{Property}[section]
\newtheorem{theorem}{Theorem}[section]
\newtheorem{rmk}{Remark}[section]

\newtheorem{coro}[theorem]{Corollary}

\newtheorem{claim}{Claim}[section]

\newrgbcolor{lightlightlightgray}{0.9 0.9 0.9}

\numberwithin{equation}{section}

\begin{document}

	\setlength{\parskip}{0cm}

\title{On the edge capacitated Steiner tree problem}
\author{Cédric Bentz$^{(1)}$, Marie-Christine Costa$^{(2)}$,\\
Alain Hertz$^{(3)}$\\ 
{\scriptsize $^{(1)}$ CEDRIC, CNAM, 292 rue Saint-Martin, 75003 Paris,
France}\\
{\scriptsize $^{(2)}$ ENSTA Paris-Tech, University of Paris-Saclay (and CEDRIC CNAM),  91762 Palaiseau Cedex, France}\\
{\scriptsize $^{(3)}$ GERAD and D\'epartement de math\'ematiques et g\'enie industriel, Polytechnique Montr\'eal, Canada}\\
}

\date{ }

\def\thefootnote{\fnsymbol{footnote}}



\maketitle

\vspace{-0.8cm}\begin{center}
\today
\end{center}

\begin{abstract}
Given a graph $G=(V,E)$ with a root $r\in V$, {\color{black}positive} capacities $\{c(e) | e\in E\}$, and {\color{black}non-negative} lengths $\{\ell(e) | e\in E\}$, the minimum-length (rooted) edge capacitated Steiner tree problem is to find a tree in $G$ of minimum total length, rooted at $r$, spanning a given subset $T\subset V$ of vertices, and such that, for each $e\in E$, there are at most $c(e)$ paths, linking $r$ to vertices in $T$, that contain $e$.
We study the complexity and approximability of the problem, considering several relevant parameters such as the number of terminals, the edge lengths and the {\color{black}minimum and maximum} edge capacities. {\color{black}For all but one combinations of assumptions regarding these parameters, we settle the question, giving a complete characterization that separates tractable cases from hard ones}. The only remaining open case is proved to be equivalent to a long-standing open problem. We also prove close relations between our problem and classical Steiner tree as well as vertex-disjoint paths problems.
\end{abstract}


\parindent=0.5cm


\vspace{-0.5cm}\section{Introduction}
\label{sec:intro}

The graphs in this paper can be directed or undirected.
 Consider a connected graph $G=(V,E)$ with a set $T \subset V$ of terminal vertices, or simply {\it terminals}, and a  length (or
  cost) function  $\ell: E \rightarrow \mathbb{Q}^+$. Let  $r \in
  V\setminus T$  be a root vertex (i.e. there is a path from $r$ to any vertex in $V$) if $G$ is directed or a special vertex called {\it root} if $G$ is undirected.
The (rooted) Steiner tree problem (\texttt{STEINER-TREE}) is to determine a directed tree $S$ in $G$, rooted at $r$, spanning all terminals of $T$ and having a minimum total length. The undirected Steiner tree problem, where one searches for a minimum-length tree spanning the terminals in an undirected graph, has been widely studied and the associated decision problem was one of Karp's 21 NP-complete problems \cite{garey,hwang,promel}. It also has many applications, as shown in \cite{cheng, du}. This problem is APX-hard \cite{bern}, but it can be solved in polynomial time when the number of terminals is fixed \cite{dreyfus,feldman,watel}, and admits constant ratio approximation algorithms otherwise \cite{byrka,robins}. There are less results about the directed version, which is a generalization of the undirected one and of {\color{black}the Set Cover problem}, and only non constant ratio approximation algorithms are known \cite{charikar,feige}. Directed problems occur  for instance in VLSI design \cite{cong} or  in multicast routing \cite{cheng}.

We consider in this paper a generalization of the (rooted) Steiner tree problem. Assume we are given a \textit{capacity} function $c: E \rightarrow \mathbb{N}^*$, where $c(e)$ is an upper bound on the number of paths containing $e$ and linking $r$ to terminals. Equivalently, for every $e=(u,v)$ in a tree $S$ rooted at $r$, the subtree of $S$ rooted at $v$ cannot contain more than $c(e)$ terminals.  Without loss of generality, we assume that $c(e) \leq K$ {\color{black}for each edge $e$}. The minimum-length capacitated (rooted) Steiner tree problem is defined as follows:\\

\noindent\textbf{Minimum-length (rooted) Capacitated Steiner Tree Problem}  (\texttt{ML-CAP-STEINER-TREE})

\noindent \emph{Input}. A connected graph $G = (V,E)$; a set $T=\{t_1,...,t_K\} \subset V$ of $K\geq 2$ terminals; a root vertex $r \in V\setminus T$; two functions on $E$: a nonnegative length function $\ell$ and a  positive capacity function $c$.

\noindent \emph{Objective}. Determine, {\color{black}if it exists},
 a minimum-length directed tree $S$ rooted at $r$, that spans all the vertices of $T$ and does not violate the capacity constraints.\\

If $G=(V,E)$ is undirected and $e=(u,v)$ is an arc of $S$, then $[u,v]$ must be an edge of $E$.
 Note that \texttt{STEINER-TREE} is the special case of \texttt{ML-CAP-STEINER-TREE} where $c(e) = K$ for all $e\in E$ {\color{black}(in this case, a feasible solution always exists)}. \texttt{ML-CAP-STEINER-TREE} appears naturally in several contexts, for example when designing a wind farm collection network \cite{hertz,pillai}, in the design of telecommunication networks \cite{lee} or in power distribution system optimization \cite{duan}.
When $\ell(e)=0$ for all $e\in E$, \texttt{ML-CAP-STEINER-TREE} turns into a decision problem, denoted by \texttt{CAP-STEINER-TREE}, and consisting of determining whether there exists or not a tree rooted at $r$, spanning all the terminals, and not violating the capacity constraints.

 When $K=n-1$,  i.e. a feasible solution is a spanning tree, \texttt{ML-CAP-STEINER-TREE} is solvable in polynomial time if  $c(e)=2$ for all $e\in E$, while it is NP-hard if  $c(e)=3$ for all $e\in E$ \cite{garey, papadimitriou}. Several authors propose models and methods based on mathematical programming to solve {\color{black}this capacitated} spanning tree problem for real-life applications such as telecommunication network design problems  \cite{bousba,uchoa,voss}. Their methods allow to solve the case where there is a positive integer demand at each vertex (instead of a unit demand as in \texttt{ML-CAP-STEINER-TREE}). In \cite{arkin,jothi}, the authors provide approximation algorithms for a variant of \texttt{ML-CAP-STEINER-TREE} where the capacities are uniform and the problem always admits a feasible solution, since it is assumed that a metric completion of the graph is available.
This paper addresses the problem where the demand is equal to 1 for each terminal vertex and $K \leq n-1$.

As will be made clear in the next sections, there are strong links between\break \texttt{ML-CAP-STEINER-TREE} and the two following famous problems, namely the minimum-length vertex-disjoint paths problem (\texttt{ML-VDISJ-PATH}) and the minimum-length edge-cost flow problem (\texttt{EDGE-COST-FLOW}).
	
\vspace{0.6cm}\noindent\textbf{Minimum-Length Vertex-Disjoint Paths Problem} (\texttt{ML-VDISJ-PATH})

\noindent\emph{Input}. A graph $G=(V,E)$; a nonnegative length function $\ell$ on $E$; $p$ disjoint vertex pairs $(s_1,s'_1),\ldots,(s_p,s'_p)$.

\noindent\emph{Objective}. Find $p$ mutually vertex-disjoint paths $\mu_1,\ldots,\mu_p$ of minimum total length  so that $\mu_i$ links $s_i$ to $s'_i$ $(i=1,\ldots,p)$.

\vspace{0.3cm}\noindent\textbf{Minimum Edge-Cost Flow Problem} (\texttt{EDGE-COST-FLOW})

\noindent\emph{Input}. A graph $G=(V,E)$; a positive integer $K$; two specified vertices $s$ and $t$; a nonnegative length function $\ell$ on $E$; a positive capacity function $c$ on $E$.

\noindent\emph{Objective}. Find a minimum-length feasible flow of $K$ units from $s$ to $t$, where the length of a flow is the sum of the lengths of the arcs/edges carrying a positive flow.

\vspace{0.4cm} When $\ell(e)=0$ for all $e\in E$, \texttt{ML-VDISJ-PATH} is known as the vertex-disjoint paths problem and will be denoted by \texttt{VDISJ-PATH}. It is NP-complete in directed and undirected graphs \cite{garey} and remains NP-complete for fixed $p$ in directed graphs \cite{fortune}, but it can be solved in polynomial time if $p$ is fixed and the graph is either undirected \cite{robertson} or a directed acyclic graph \cite{fortune}. The NP-hardness results for \texttt{VDISJ-PATH} apply to \texttt{ML-VDISJ-PATH} as well, but for this latter problem the complexity is unknown in the case where $p$ is fixed and the graph is undirected. However, a polynomial-time probabilistic algorithm for $p=2$ has been recently presented in \cite{bjorklund}.

For any graph theoretical terms not defined here, the reader is referred to \cite{west}.  We use the term {\em path} both for a chain when the graph is undirected,  and for a directed path when the graph is directed, i.e. when it is a digraph. If the graph is directed, recall that, in the definition of \texttt{ML-CAP-STEINER-TREE}, $r$ is assumed to be a root vertex. This is a trivial necessary condition for the existence of a feasible solution and can be easily checked.  Since all trees studied in this paper are directed from $r$ towards the terminals, we use the term {\em tree} instead of directed tree. For a vertex $v$ in a tree $S$, we denote by $S(v)$ the subtree of $S$ rooted at $v$.
For a subgraph $G'=(V',E')$ of $G$, we indifferently denote by $\ell(G')$ or $\ell(E')$ the sum of the lengths of the arcs/edges in $G'$.
 Also, for $e\in E$,  a rooted tree $S$ in $G$, and  two vertices $u$ and $v$  such that $v$ is a descendant of $u$ in $S$, we say that $u$ is $e$\textit{-linked} (resp. $\bar{e}$\textit{-linked}) to $v$ in $S$ if $e$ belongs (resp. does not belong) to the path  $\mu_{uv}$ from $u$ to $v$ in $S$. Similarly, when we write that $r$ is $\bar{e}$-linked to a subset $T'$ of terminals in $S$, this means that $e$ does not belong to the paths in $S$ that link $r$ to the terminals  of $T'$. The capacity constraints therefore impose that, for all $e\in E$, $r$ is $e$-linked to at most $c(e)$ terminals in any {\color{black}feasible} solution $S$ to an \texttt{ML-CAP-STEINER-TREE} {\color{black}instance}. Equivalently,
$S(v)$ contains at most $c(e)$ terminals for all $e=(u,v)$ in $S$.

The next section gives an overview of our results concerning \texttt{ML-CAP-STEINER-TREE} and explains how the paper is organized.

\vspace{-0.2cm}\section{Overview of the results}
 \label{sec:results}
\vspace{-0.2cm}
In this section we show that our results provide a complete characterization of the complexity of \texttt{ML-CAP-STEINER-TREE} that allows us to distinguish beween easy and hard cases of  the problem for  digraphs, directed acyclic graphs (called DAGs) and  undirected graphs. Notice that any undirected instance of \texttt{ML-CAP-STEINER-TREE} can be transformed into a directed one by replacing each edge by two opposite arcs having the same length and capacity. Hence, any positive result (existence of a polynomial-time algorithm or approximation result) for directed graphs is also true for undirected graphs, while any negative result for undirected graphs (NP-hardness or non-approximability result) is also true for directed graphs.

\enlargethispage{0.5cm}{\color{black}Apart from the assumption on the graph itself (undirected, directed or directed without circuits)}, the following parameters are considered: the number $K$ of terminal vertices, the {\color{black}minimum and maximum} edge capacities, and the edge lengths. More precisely, $K$ can be fixed or not; {\color{black}the minimum and maximum edge capacities can be non depending on $K$ (equal to 1 or not), they can be greater than or equal to $K-\kappa$ ($1 \leq \kappa \leq K-1$), and they can be equal (uniform capacity) or not}; the edge lengths can be all equal to $0$, all equal to a positive value (i.e. uniform), or non uniform. We settle all cases except one, namely the undirected case with uniform capacity and fixed $K \geq 3$, but we prove that \texttt{ML-CAP-STEINER-TREE} is then equivalent to \texttt{ML-VDISJ-PATH} {\color{black}in undirected graphs with fixed $p$}, whose complexity is a long-standing open problem {\color{black}in this case} \cite{kobayashi}.

\enlargethispage{1.5cm}Our results are summarized in four tables. Each line of each table corresponds to a specific case of  \texttt{ML-CAP-STEINER-TREE} and refers to the theorem  where the case is settled.
The first table contains  results that are valid for the three types of graphs, while the next three tables contain results that are specific to digraphs, undirected graphs, and DAGs, respectively. In these tables, we denote by $\rho$ the best possible approximation ratio for \texttt{STEINER-TREE}, and by $\rho'$ the best possible approximation ratio for \texttt{ML-VDISJ-PATH} with a fixed number of source-sink pairs.

The three trees drawn in Figure 1 provide another picture of the possible cases for the three types of graphs (digraphs, DAGs and undirected graphs). The numbers assigned to the leaves of these trees refer to the corresponding rows in the tables.  The values of the three parameters appear on the branches and each branching node corresponds to a partition of the possible cases: the value on a branch excludes the values on the branches to the left. For instance, in undirected graphs, the capacities can be either uniform equal to 1, or at least $K-1$, or uniform of value at least 2 and at most $K-2$, or, finally, any capacities not yet considered.

{\color{black}Moreover, if a leaf corresponds to a branch where the values of some parameters are unspecified, then this means that the associated result holds even in the {\color{red}most} general case (if it is a positive, i.e., tractability result) or in the {\color{red}most} specific case (if it is a negative, i.e., hardness result) with regard to the unspecified values. For instance, the NP-hardness result associated with Leaf 7 holds even if $K$ is fixed and if all lengths are 0 (since neither the value of $K$ nor the lengths appear on this branch), and the result associated with Leaf 11 holds for any lengths and any capacities (since only the assumption on $K$ being fixed appears on this branch).}

{\color{black}Therefore}, for digraphs, the branch ``any capacity" includes the case of uniform capacities between 2 and $K-2$  for $K$ fixed (or not).  Concerning the last line of Table 3, if the capacity is uniform and $K$ is fixed, then there exists some constant $\kappa$ such that all capacities are equal to $K-\kappa$: hence, in the tree dealing with undirected graphs in Figure \ref{fig:arbres}, the branch ``any capacity'', which leads to Case 8 of Table \ref{tabUndirectedGraphs}, excludes the case where $K$ is fixed.

\begin{table}[h!]
	\begin{center}
		\begin{tabular}
			{|p{10pt}|p{150pt}|p{150pt}|p{50pt}|} \hline
			\small{}  & \small{Condition} & \small{Complexity}& \small{Theorem} \\
			
			\hline \small{1}  & \small{Unit capacities} & \scriptsize{Polynomial} & \scriptsize{Theorem \ref{th:CSTP-cap1}}\\
			\hline  \small{2}  & \small{$K=2$} & \scriptsize{Polynomial} & \scriptsize{Theorem \ref{th:large-capas-fixedK}}\\
			\hline  \small{3}  & \small{Capacities $\geq K-\kappa$, for any constant $\kappa \geq 0$} & \scriptsize{NP-hard, even with lengths 1, even with uniform capacities} & \scriptsize{Theorem \ref{th:unit-cost-Steiner}}\\
			\hline  \small{4}  & \small{Capacities $\geq K-1$} & \scriptsize{Polynomial with lengths 0 (\texttt{CAP-STEINER-TREE}), and $(1+\rho)$-approximable otherwise} & \scriptsize{Theorem \ref{th:large-capas-non-fixedK}}\\
			\hline  \small{5}  & \small{Capacities $\geq K-1$, for fixed $K$} & \scriptsize{Polynomial} & \scriptsize{Theorem \ref{th:large-capas-fixedK} }\\
			\hline
		\end{tabular}
		\caption{General results for \texttt{ML-CAP-STEINER-TREE} in digraphs, DAGs, and undirected graphs.} \label{tabGeneralResults}
	\end{center}\end{table}
	
	\begin{table}[H]
		\begin{center}
			\begin{tabular}
				{|p{10pt}|p{150pt}|p{150pt}|p{50pt}|} \hline  \small{} &
				\small{Condition} & \small{Complexity}& \small{Theorem} \\
				\hline  \small{6}  &  \small{$K \geq 3$ (fixed or not)} & \scriptsize{NP-complete even if all lengths are 0 (\texttt{CAP-STEINER-TREE}), and even if the minimum capacity $c_{\min}$ and the maximum capacity $c_{\max} \geq c_{\min}$ are any fixed constants, with $c_{\min} \in \{1, \dots, K-2\}$} and $c_{\max} \geq 2$ & \scriptsize{Theorem \ref{th:CSTP-directed-unif-K-fixed}}\\
				\hline
			\end{tabular}
			\caption{Results for \texttt{ML-CAP-STEINER-TREE} in digraphs.} \label{tabDigraphs}
		\end{center}\end{table}

		\begin{table} [H]
			\begin{center}	
				\begin{tabular}
					{|p{10pt}|p{150pt}|p{150pt}|p{50pt}|} \hline
					\small{}  & \small{Condition} & \small{Complexity}& \small{Theorem} \\
					\hline \small{7}  &  \small{Non uniform capacities and $K \geq 3$ (fixed or not)} & \scriptsize{NP-complete even if all lengths are 0 (\texttt{CAP-STEINER-TREE}), and even if the minimum capacity $c_{\min}$ and the maximum capacity $c_{\max} > c_{\min}$ are any values, with $c_{\min} \in \{1, \dots, K-2\}$} & \scriptsize{Theorem \ref{th:CSTP-undirected-non-unif}}\\
					\hline \small{8}  &  \small{Uniform capacity (non unit and not depending on $K$)} & \scriptsize{NP-complete even if all lengths are 0 (\texttt{CAP-STEINER-TREE}), and even if the uniform capacity is any value $\geq 2$ not depending on $K$} & \scriptsize{Theorem \ref{th:unif-DAG-SAT}}\\
					\hline \small{9}  &  \small{Uniform capacity equal to $K-\kappa$, for any constant $\kappa \geq 0$} & \scriptsize{Polynomial if all lengths are 0 (\texttt{CAP-STEINER-TREE}), and $(\rho+\rho')$-approximable otherwise} & \scriptsize{Theorems \ref{th:undirected-unif-fixedK} and \ref{th:CSTP-large-capa-unif}}\\
					\hline \small{10}  &  \small{Uniform capacity and fixed $K \geq 3$} & \scriptsize{Equivalent to \texttt{ML-VDISJ-PATH} with fixed $p$, and hence open} & \scriptsize{Theorem \ref{th:undirected-unif-fixedK}}\\
					\hline
				\end{tabular}
				\caption{Results for \texttt{ML-CAP-STEINER-TREE} in undirected graphs.} \label{tabUndirectedGraphs}
			\end{center}\end{table}
			
			\begin{table}[H]
				\begin{center}\begin{tabular}
						{|p{10pt}|p{150pt}|p{150pt}|p{50pt}|} \hline \small{}  &
						\small{Condition} & \small{Complexity}& \small{Theorem} \\
						\hline \small{11}  & \small{Fixed $K$} & \scriptsize{Polynomial} & \scriptsize{Theorem \ref{th:DAG-fixedK}}\\
						\hline \small{12}  & \small{Non unit capacities not depending on $K$} & \scriptsize{NP-complete even if all lengths are 0 (\texttt{CAP-STEINER-TREE}), and even if the capacity is uniform and takes any value $\geq 2$ not depending on $K$} & \scriptsize{Theorem \ref{th:unif-DAG-SAT}}\\
						\hline \small{13}  & \small{Capacities larger than $K-\kappa$, for any constant $\kappa \geq 0$} & \scriptsize{Polynomial if all lengths are 0 (\texttt{CAP-STEINER-TREE}), and $(1+\rho)$-approximable otherwise} & \scriptsize{Theorem \ref{th:CSTP-large-capa-non-unif-DAGs}}\\
						\hline
					\end{tabular}
					\caption{Results for \texttt{ML-CAP-STEINER-TREE} in DAGs.} \label{tabDAGs}
				\end{center}\end{table}

We describe in Section 3 the structure of optimal solutions to an \texttt{ML-CAP-STEINER-TREE} instance. Section 4 is devoted to relations between \texttt{ML-CAP-STEINER-TREE} and \texttt{ML-VDISJ-PATH}. We prove in Section 5 some NP-hardness results for the general case, while special cases where the number $K$ of terminals is fixed, or where all capacities are almost equal to $K$, are studied in Sections 6 and 7.

\begin{figure} [h!]
\begin{center}
\includegraphics[scale=0.8]{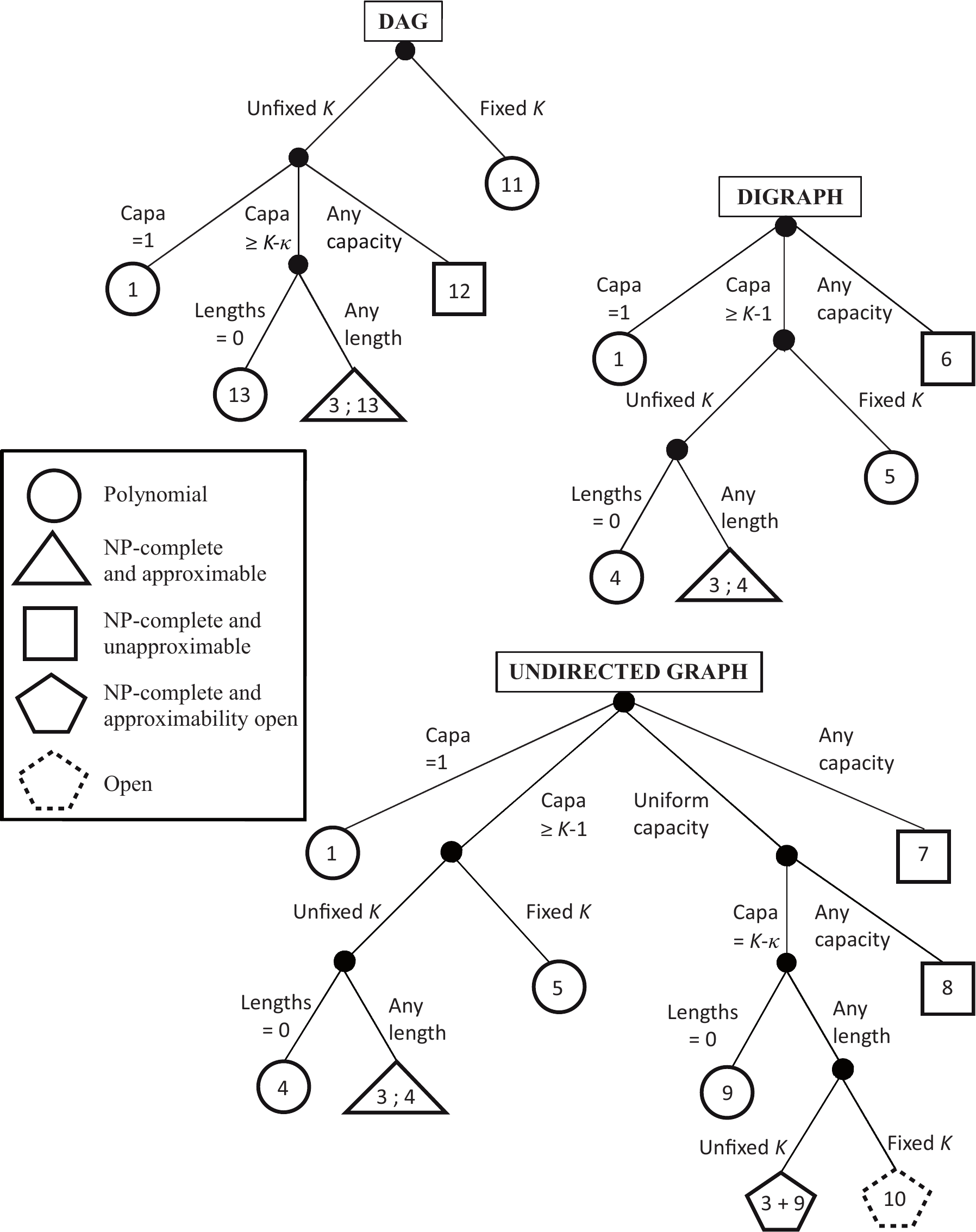}
\caption{Results for \texttt{ML-CAP-STEINER-TREE} (\texttt{CAP-STEINER-TREE} if all lengths\break are 0) in digraphs, DAGs and undirected graphs.}
\label{fig:arbres}
\end{center}
\end{figure}

\section{Structural properties of optimal solutions}
\label{sec:struc}

We can assume, without loss of generality, that there is a bijection between the set of 1-degree vertices (leaves) in $V\setminus \{r\}$ and $T$. Indeed, if $t \in T$ is not a leaf, we can add a new terminal vertex $t'$ and an edge $[t,t']$ (or an arc $(t,t')$) with capacity $1$ and length $0$, and replace $t$ by $t'$ in $T$. Moreover, if there is a leaf $v\notin T\cup \{r\}$ in $G$, then $v$ can be removed from $G$ since the removal of $v$ from a solution $S$ to an \texttt{ML-CAP-STEINER-TREE} instance gives a solution $S'$ which is at least as good as $S$.

A solution $S$ (if any) to an \texttt{ML-CAP-STEINER-TREE} instance is a tree rooted at $r$, and defines $K$ paths from $r$ to the $K$ terminals. The vertices with degree  at least 3 in $S \backslash \{r\}$ are called {\em junction vertices}.

To each junction vertex $v$, we associate the set $T_v\subseteq T$ of terminals in the subtree $S(v)$ rooted at $v$. Moreover, for an arc $e=(u,v)$ in $S$,  $|T_v|$ is the number of terminals to which $r$ is $e$-linked in $S$. If there is no directed path linking two vertices $v$ and $w$ in $S$, then $T_v\cap T_w= \emptyset$, otherwise $S$ would contain a cycle. 

 Given a tree $S$ {\color{black}spanning a set $T$ of terminals}, its \textit{skeleton} is the tree obtained from $S$ by iteratively contracting vertices $v \notin  T\cup \{r\}$ with exactly one incoming arc $(u,v)$ and exactly one outgoing arc $(v,w)$ (i.e., the path $(u,v,w)$ is replaced by an arc $(u,w)$).  {\color{black}This means that} there is an arc $(u,v)$ in the skeleton of $S$ if and only if there is a path from $u$ to $v$ in $S$, each internal vertex of this path being of degree 2 in $S$. 
 When all capacities are 1, the skeleton of a feasible solution is a star, since the root is the only possible vertex with degree $\geq 2$ in this skeleton. We now prove some properties which will be useful later.

\begin{prop}
\label{prop:numbvert}
The skeleton of an {\color{black}inclusion-wise minimal} tree $S$ {\color{black}rooted at $r$ and spanning $K$ terminals (all of degree 1)} contains at most $2K+1-d_r$ vertices, where $d_r$ is the degree of root $r$ in $S$.
\end{prop}

\begin{proof}
Let $n_J$ be the number of junction vertices in the skeleton $R$ of $S$. Clearly, $R$ contains $n_R=K+1+n_J$ vertices and $n_R-1$ edges. Since the sum of the degrees of all vertices in $R$ is $2(n_R-1)=2K+2n_J$, we have $2K+2n_J\geq K +d_r+3n_J$,
which implies $n_J\leq K -d_r$ and $n_R\leq 2K+1-d_r$.
\end{proof}

\begin{prop}
\label{prop:pathlength}
{\color{black}Given an inclusion-wise minimal tree $S$ rooted at $r$ and spanning $K$ terminals (all of degree 1)}, the path with minimum number of vertices from root $r$ to a terminal in the skeleton of $S$ contains at most $O(\log (K))$ vertices.
\end{prop}

\begin{proof}
Let $R$ be the skeleton of $S$, {\color{black}$n_R$} its number of vertices, and $l_{\min}$ the minimum number of vertices on a path from $r$ to a terminal in $R$.
\begin{itemize}
	\vspace{-0.2cm}\item If $r$ has degree 1 in $S$, then $R$ contains one vertex at levels 1 and 2, and at least $2^{i-2}$ vertices at levels $i=3, \dots, l_{\min}$. Hence, $n_R\geq 2+\sum_{i=1}^{l_{\min}-2}2^{i}=2^{l_{\min}-1}$, which implies
$l_{\min} \leq \log_2(n_R)+1$.
\vspace{-0.2cm}\item If  $r$ has degree at least 2 in $S$, then $R$ contains at least $2^{i-1}$ vertices at levels $i=1, \dots, l_{\min}$. Hence, $n_R\geq \sum_{i=0}^{l_{\min}-1}2^{i}=2^{l_{\min}}-1$, which implies $l_{\min} \leq \log_2(n_R+1)$.
\end{itemize}
\vspace{-0.2cm}In both cases, it follows from Property \ref{prop:numbvert} that $l_{\min} = O(\log(K))$.
\end{proof}

{\color{black}Notice that, if $S$ is a complete binary tree, then $l_{\min} = \Omega(\log(K))$: therefore, up to a constant factor, the bound in the previous property cannot be improved.}

Given a graph $G$ with a root $r$ and $K$ terminals, a {\em potential skeleton} in $G$ is defined as a tree $P$, rooted at $r$, {\color{black}spanning the $K$ terminals}, and such that the only vertices without outgoing arc are the $K$ terminals, while any other vertex, except possibly $r$, has degree at least 3 in $P$ (and hence in $G$). While the skeleton of a solution to an \texttt{ML-CAP-STEINER-TREE} {\color{black}instance} is a potential skeleton, the reverse is not necessarily true, as illustrated in Figure \ref{fig:potential}.  If we select the left arc incident to $r$ in the first potential skeleton of the figure,  then there are no three vertex-disjoint paths from the left neighbor of $r$ to the terminals in $G$.

\begin{figure} [H]
	\begin{center}		\includegraphics[scale=0.85]{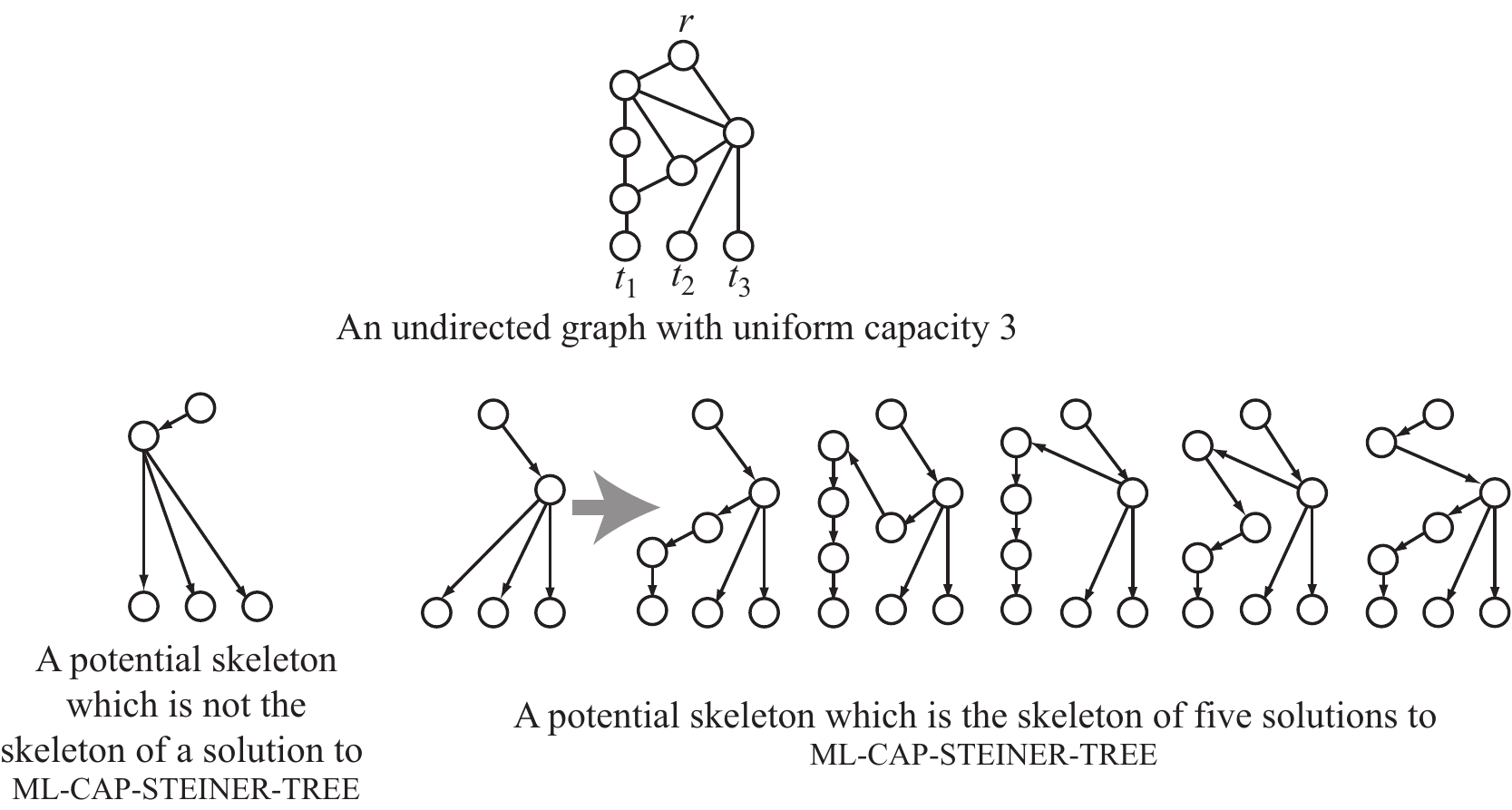}
		\caption{Potential skeletons in a graph $G$.}
		\label{fig:potential}
	\end{center}
\end{figure}

\begin{prop}
\label{enumskelet}
{\color{black}Given a graph $G$ with $n$ vertices, $K$ terminals, and a root vertex $r$, it is possible to enumerate in $O(n^{K-1}K^{O(K)})$ time all potential skeletons of inclusion-wise minimal trees rooted at $r$ and spanning the $K$ terminals in $G$.}
\end{prop}

\begin{proof}
As shown in the proof of Property \ref{prop:numbvert}, the skeleton of such a tree contains at most $K-1$ junction vertices. There are $O(n^{K-1})$ ways of choosing at most $K-1$ junction vertices, and, for each such choice, it follows from Cayley's formula that there are at most $(2K)^{2K-2}$ different labelled trees containing only $r$, the $K$ terminals, and the chosen junction vertices. We can orient the edges of every labelled tree from the root $r$ towards the other vertices, which takes $O(K)$ time per tree, and reject the labelled rooted trees that do not satisfy the definition of a potential skeleton. The whole procedure therefore takes
$O(n^{K-1}K^{O(K)})$ time.
\end{proof}

\section{Links with vertex-disjoint paths problems}
\label{sec:linkpath}

We detail in this section several links between \texttt{ML-CAP-STEINER-TREE} and some vertex-disjoint paths problems. We begin with a simple complexity result in the case of unit capacities. In this case, an optimal solution to \texttt{ML-CAP-STEINER-TREE} necessarily consists of $K$ vertex-disjoint paths with minimum total length, each one linking $r$ to a terminal, and we obtain the following theorem.

\begin{theorem}
\label{th:CSTP-cap1}
 \texttt{ML-CAP-STEINER-TREE} is polynomial-time solvable if $c(e) = 1$  $\forall e \in E$.
\end{theorem}
\begin{proof}
Assume the input graph $G$ is directed, and let us add to $G$ a new vertex $s$ and an arc $(t_k,s)$ of length 0 and capacity 1 for each terminal $t_1,\ldots,t_k$.  Solving \texttt{ML-CAP-STEINER-TREE} then amounts to finding $K$ internally vertex-disjoint paths from $r$ to $s$, with minimum total length. It is well-known that this can be done in polynomial time, but we briefly recall how. We consider the graph $H$ obtained from $G$ by replacing each vertex $v \notin  \{r,s,t_1,\ldots,t_k\}$ by an arc $(v',v'')$ of length 0, and each arc $(v_1,v_2)$ (resp. $(r,v),(v,t_i), i=1,...,k$) by an arc $(v_1'',v_2')$ (resp. $(r,v'),(v'',t_i), i=1,...,k$) having the same length as the original one. All capacities are set equal to 1. It is then sufficient to determine a minimum-cost flow of $k$ units from $r$ to $s$ in $H$ by using any min-cost flow algorithm \cite{gondran}. Recall that, if the graph $G$ is undirected, we can transform it into a directed one by replacing each edge by two opposite arcs. In this case, only one of two opposite arcs associated to an edge carries a positive flow in the solution.
\end{proof}

\noindent The following problem is a generalization of \texttt{ML-VDISJ-PATH}.\\

\noindent\textbf{Minimum-Length Labelled Vertex-Disjoint Paths Problem}  (\texttt{ML-LAB-VDISJ-PATH})

\noindent\emph{Input}. A graph $G=(V,E)$; an integer $k \geq 1$; a nonnegative length function $\ell$ on $E$; a label $\lambda(e)\in \{1,\dots,k\}$ on every $e\in E$; $p$ disjoint vertex  pairs ($s_i$, $s'_i$), each one being associated with a set $L_i \subseteq \{1,\dots,k\}$ of labels.

\noindent\emph{Objective}: find $p$ mutually vertex-disjoint paths $\mu_1,\ldots,\mu_p$ of minimum total length  so that $\mu_i$ links $s_i$ to $s'_i$ and all labels on $\mu_i$ belong to $L_i$ $(i=1,\ldots,p)$.
\\

 \noindent When $\ell(e)=0$ for all $e\in E$, \texttt{ML-VDISJ-PATH} (resp. \texttt{ML-LAB-VDISJ-PATH}) turns into a decision problem, denoted by \texttt{VDISJ-PATH} (resp. \texttt{LAB-VDISJ-PATH}).
Notice that \texttt{ML-VDISJ-PATH} is the special case of \texttt{ML-LAB-VDISJ-PATH} where $L_i = \{1,\dots,k\}$ for $i=1,\ldots,p$. We now show several links between \texttt{ML-CAP-STEINER-TREE} and some variants of \texttt{ML-VDISJ-PATH} and \texttt{ML-LAB-VDISJ-PATH}.

\begin{theorem}\label{th:reductionFromMVDP}
\texttt{ML-VDISJ-PATH} with $p$ source-sink pairs is polynomially reducible to \texttt{ML-CAP-STEINER-TREE} with $p(p+1)/2$ terminals.
\end{theorem}

\begin{proof}
Assume first that the input graph $G=(V,E)$ of the  \texttt{ML-VDISJ-PATH} instance is undirected. Let $G'=(V',E')$ be defined as follows: $V'$ is obtained  by adding to $V$ a vertex $r$ and $K=p(p+1)/2$ terminals $t_{i_j}$, $1\leq j\leq i \leq p$; $E'$ is obtained from $E$ by adding an edge of capacity $i$ and length 0 between $r$ and  every $s_i$, $i=1,\ldots,p$, as well as edges of capacity 1 and length 0 between $s'_i$ and every $t_{i_j}$, $1\leq j\leq i \leq p$. The edges of $E$ keep their original length, while their capacity is fixed to $p$. We prove that solving \texttt{ML-VDISJ-PATH} in $G$ is equivalent to solving \texttt{ML-CAP-STEINER-TREE} in $G'$.  The construction of $G'$ from $G$ is illustrated in Figure \ref{fig:mndpcstp} for $p=3$, with the pair $(c(e),\ell(e))$ on every $e\in E'$.

\begin{figure} [htbp]
\begin{center}
\includegraphics[scale=0.8]{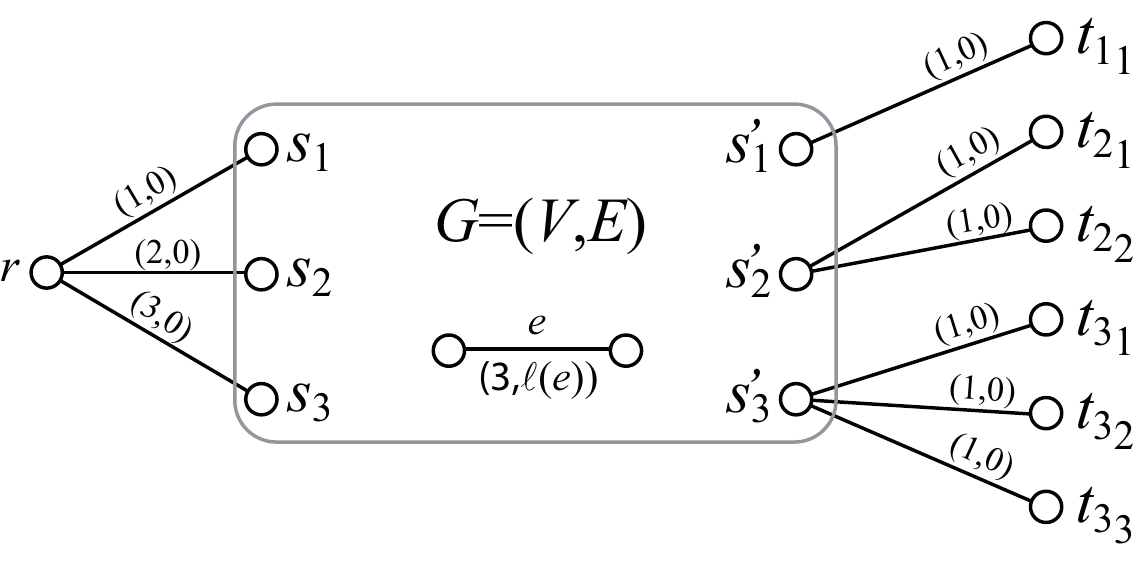}
\caption{From \texttt{ML-VDISJ-PATH}  in $G=(V,E)$ with $p=3$ to \texttt{ML-CAP-STEINER-TREE} in $G'=(V',E')$ with $K=6$.}

\label{fig:mndpcstp}
\end{center}
\end{figure}

Given a solution $S$ to \texttt{ML-VDISJ-PATH} in $G$, one can get a solution $S'$ to\break \texttt{ML-CAP-STEINER-TREE} in $G'$ of same total length by orienting all paths from $s_i$ to $s'_i$, $i=1,\ldots,p$, and then adding the $p$ arcs $(r,s_i)$, as well as the $p(p+1)/2$ arcs incident to the terminals.

Now, assume there is a solution $S'$ for \texttt{ML-CAP-STEINER-TREE} in $G'$. Since there are $p(p+1)/2$ terminals while the sum of the capacities of the  edges incident to $r$ is precisely this amount, we know that $r$ is $(r,s_i)$-linked to exactly $i$ terminals in $S'$, $i=1,\ldots,p$. In particular,
$r$ is $(r,s_p)$-linked to $p$ terminals, and these are necessarily $t_{p_1},\ldots,t_{p_p}$, otherwise $S'$ would contain a cycle. Using the same argument, with $i$ decreasing from $p$ to $1$, we get that $r$ is $(r,s_i)$-linked to $t_{i_1},\ldots,t_{i_i}$.  Notice that all paths from $r$ to $t_{i_j}$, $j=1,...,i$, use the same sub-path from $s_i$ to $s'_i$.  Hence, by removing from $S'$ all arcs incident to $r$ and to the terminals, we get a solution $S$ to \texttt{ML-VDISJ-PATH} with same total length.

The proof for digraphs is obtained by replacing ``edge" by ``arc" in the construction of $G'$.
\end{proof}

\begin{theorem} \label{cor:2MVDP-unif}
Given two integers $K$ and $c$ with $K \geq 4$ and $2 \leq c \leq K-2$, \texttt{ML-VDISJ-PATH} with $p=2$ source-sink pairs is polynomially reducible to\break \texttt{ML-CAP-STEINER-TREE} with $K$ terminals and uniform capacity $c$.
\end{theorem}

\begin{proof}
The proof is similar to the previous one. The main difference is the definition of $G'=(V',E')$. In the undirected case, $V'$ is obtained by adding to $V$ two vertices $r$ and $v$ and $K$ terminals $t_1,\ldots,t_K$; $E'$ is obtained from $E$ by adding the edges $[r,v], [r,s_2], [v,s_1], [v,t_1], [s'_1,t_2], [s'_2,t_i]$ for $i=3,\ldots,c+2$, and $[r,t_i]$ for\break $i=c+3,\ldots,K$. The edges of $E$ keep their original length while those in $E'\setminus E$ have length 0. All capacities are set equal to $c$. We then prove that \texttt{ML-VDISJ-PATH} on $G$ is equivalent to \texttt{ML-CAP-STEINER-TREE} in $G'$ in a similar way as in the previous theorem. The only path that goes from $r$ to $t_1$ contains $rv$, and hence the remaining capacity on this edge is $c-1$: this implies that the paths from $r$ to the $c$ terminals adjacent to $s'_2$ must contain $rs_2$, and the rest of the proof is unchanged. The
proof for digraphs is obtained by adding arcs instead of edges to obtain $G'$.
\end{proof}

\begin{theorem} \label{cor:pMVDP-unif}
\texttt{ML-VDISJ-PATH} with $p \geq 2$ source-sink pairs is polynomially reducible to \texttt{ML-CAP-STEINER-TREE} with $p^2$ terminals and uniform capacity $p$.
\end{theorem}
\begin{proof}
Again, the proof is similar to the one  of Theorem \ref{th:reductionFromMVDP}. In this case, $G'=(V',E')$ is constructed as follows. $V'$ is obtained  by adding to $V$ a vertex $r$, $p-1$ vertices $v_1,\ldots,v_{p-1}$ and $p^2$ terminals $t_{i_j}$ with $1\leq i,j \leq p$; $E'$ is obtained from $E$ by adding the edges $[r,s_p]$, $[r,v_i]$ and $[v_i,s_i]$ for $i=1,\ldots,p-1$, as well as edges between $s'_i$ and every $t_{i_j}$ with $1\leq j\leq i \leq p$ and edges between $v_i$ and every $t_{i_j}$ with $1\leq i< j \leq p$. The edges of $E$ keep their original length while those in $E'\setminus E$ have length 0. All capacities are set equal to $p$. We then prove that \texttt{ML-VDISJ-PATH} on $G$ is equivalent to \texttt{ML-CAP-STEINER-TREE} in $G'$ in a similar way as in Theorem \ref{th:reductionFromMVDP}. Notice that, in this case, given any solution $S'$ to \texttt{ML-CAP-STEINER-TREE} in $G'$, $r$ is necessarily $(r,v_i)$-linked to terminals $t_{i_j}$ ($j=1,\ldots,p$) for all $i=1,\ldots,p-1$, and $r$ is $(r,s_p)$-linked to terminals $t_{p_j}$ ($j=1,\ldots,p$).
\end{proof}

\vspace{-0.4cm}\begin{rmk}
{\em The results stated in Theorems \ref{th:reductionFromMVDP}, \ref{cor:2MVDP-unif} and \ref{cor:pMVDP-unif} are also valid for \texttt{ML-VDISJ-PATH} and  \texttt{ML-CAP-STEINER-TREE} with strictly positive lengths, since the arcs or edges added to $G$ in order to obtain $G'$ can have arbitrary lengths. Indeed,
the total length of a solution $S$ to \texttt{ML-VDISJ-PATH} will then differ from the total length of the corresponding solution $S'$ to \texttt{ML-CAP-STEINER-TREE} by a value equal to the total length of the added arcs or edges.}
\end{rmk}

\vspace{-0.1cm}We next show that, when the number $K$ of terminals is fixed, \texttt{ML-CAP-STEINER-TREE} is polynomially reducible to \texttt{ML-LAB-VDISJ-PATH}.

\begin{theorem}\label{th:reductionToMLVDP}
When $K\geq 1$ is fixed, \texttt{ML-CAP-STEINER-TREE}  can be reduced in polynomial time to \texttt{ML-LAB-VDISJ-PATH} with a fixed number of source-sink pairs.
\end{theorem}

\vspace{-0.2cm}\begin{proof}
We first consider the undirected case. Let $I$ be an instance of\break \texttt{ML-CAP-STEINER-TREE} in a graph $G$ containing $K$ terminals. It follows from Property \ref{enumskelet} that the set of potential skeletons of {\color{black}optimal} solutions to $I$ can be enumerated in  $O(n^{K-1})$ time (since $K$ is a constant), where $n$ is the number of vertices in $G$.

To every such potential skeleton $S$, we associate a graph $G'=(V',E')$ constructed as follows, in order to deal with vertex-disjoint paths (and not internally vertex-disjoint paths). For each arc $(u,v)$ of $S$, we create a copy $u_v$ of $u$ and a copy $v_u$ of $v$; hence, every vertex $v$ of $S$ is replaced in $G'$ by $d_v$ copies of $v$, where $d_v$ is the degree of $v$ in $S$. All vertices of $G$ that do not appear in $S$ are also put in $V'$ (with only one copy of each). For each edge $[u,v]$ of $G$ we put an edge {\color{black}of same length} in $G'$ between each copy of $u$ and each copy of $v$. This construction is illustrated in Figure \ref{fig:mlndp}.

We then create a source-sink pair  $(u_v,v_u)$ in $G'$ for all arcs $(u,v)$ in $S$. From Property \ref{prop:numbvert}, $S$ contains at most $2K$ vertices, and there are therefore at most $2K-1$ such pairs. For each $v \in V$, let $K_v$ denote the number of terminals in the subtree $S(v)$ of $S$ rooted at $v$, and let $L_{uv}$ be the set of labels associated with the source-sink pair $(u_v,v_u)$. We set
$L_{uv}=\{K_v,\ldots,K\}$. In the example of Figure \ref{fig:mlndp}, we  have
$L_{rb}=\{3\}, L_{be}=\{2,3\}$, and $L_{bt_1}=L_{et_2}=L_{et_3}=\{1,2,3\}$. The label $\lambda(e)$ associated with an edge $e$ in $G'$ is the capacity of the corresponding edge in $G$.

An optimal solution to \texttt{ML-LAB-VDISJ-PATH} in $G'$ (if any) corresponds to a minimum-length solution to \texttt{ML-CAP-STEINER-TREE} in $G$ having $S$ as skeleton.
Since we enumerate all potential skeletons, the best solution to \texttt{ML-CAP-STEINER-TREE} obtained during this enumeration is an optimal solution for $I$.

\noindent The proof is similar for digraphs, by replacing edge by arc.
\end{proof}

\begin{figure} [htbp]
	\begin{center}
		\includegraphics[scale=0.58]{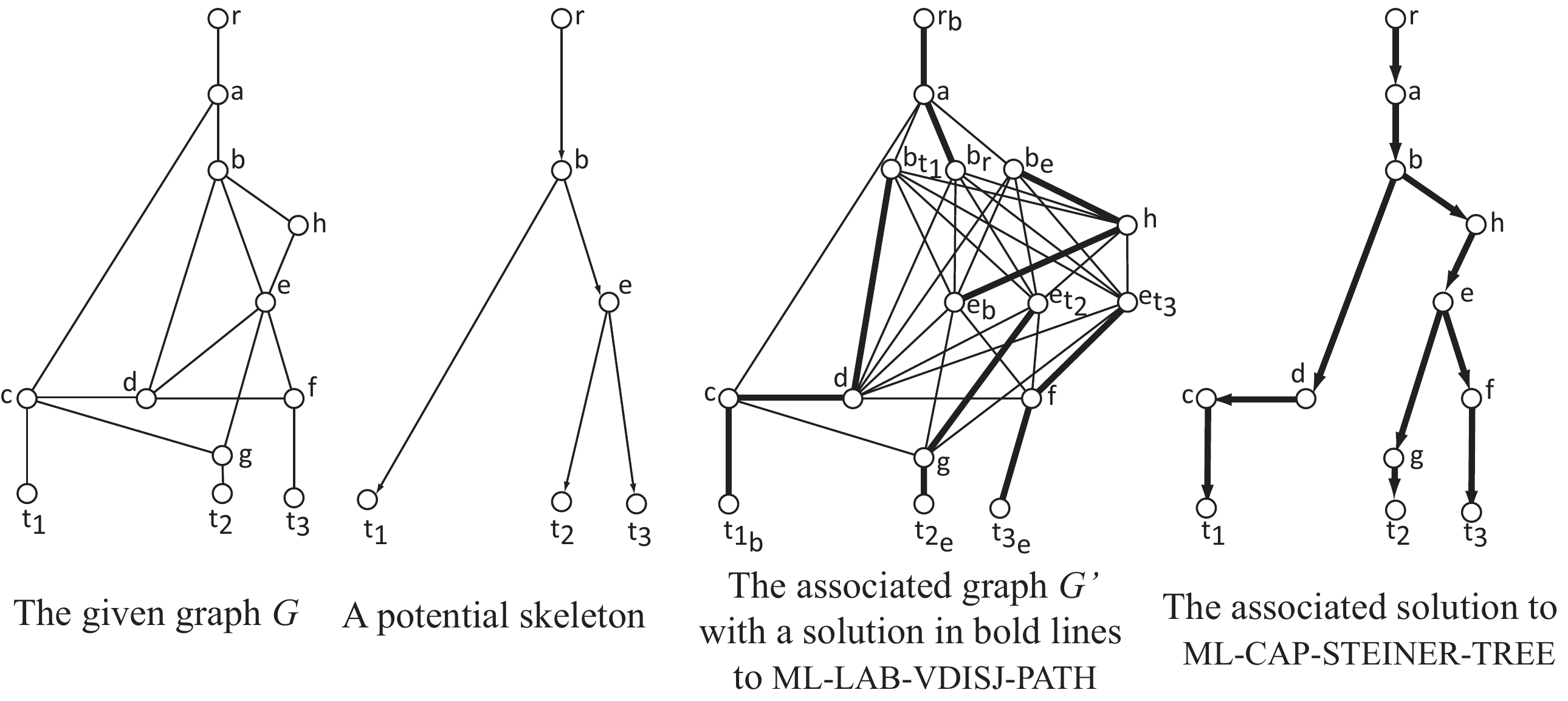}
		\caption{From \texttt{ML-CAP-STEINER-TREE} in $G=(V,E)$ to \texttt{ML-LAB-VDISJ-PATH} in $G'=(V',E').$ {\color{black}(The length of $[b,e]$ is equal to 3, and all other lengths are equal to 1.)}}	
		\label{fig:mlndp}
	\end{center}
\end{figure}

\begin{theorem} \label{cor:MCSTPuniftoMVDP}
When $K\ge 1$ is fixed, \texttt{ML-CAP-STEINER-TREE} (resp. \texttt{CAP-STEINER-TREE}) with uniform capacity is polynomially reducible to \texttt{ML-VDISJ-PATH} (resp. \texttt{VDISJ-PATH}) with a fixed number of source-sink pairs.
\end{theorem}

\begin{proof}
The proof is similar to the proof of Theorem \ref{th:reductionToMLVDP}. However, since the capacities are all equal to a constant $c$, we do not have to use labels. Consider any potential skeleton $S$: if $r$ has at least one successor $v$ such that the number of terminals in the subtree $S(v)$ of $S$ rooted at $v$ is strictly larger than $c$, then $S$ can be rejected since it cannot correspond to the skeleton of a tree satisfying the capacity constraints. Otherwise, an optimal {\color{black}(resp. a feasible)} solution to \texttt{ML-VDISJ-PATH} is a minimum-length {\color{black}(resp. a feasible)} solution to \texttt{ML-CAP-STEINER-TREE} having $S$ as skeleton.
\end{proof}

\begin{figure} [htbp]
	\begin{center}
		\includegraphics[scale=1.0]{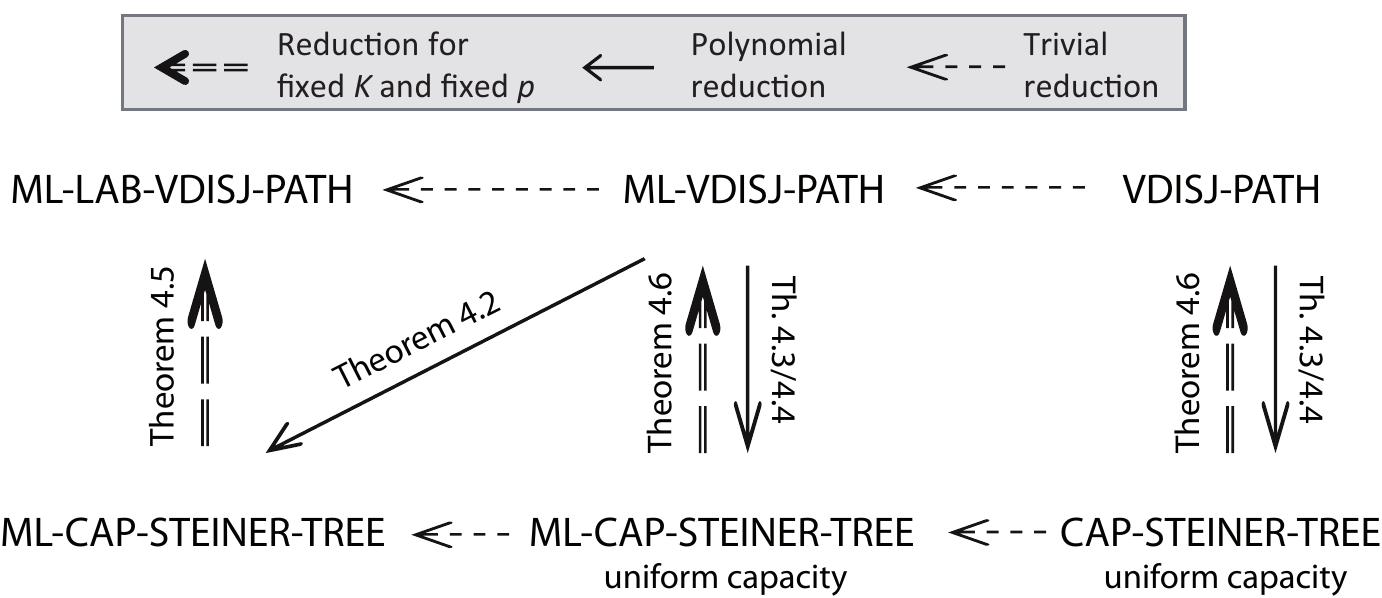}
		
		\caption[Polynomial-time reductions between capacitated Steiner tree and vertex-disjoint paths problems.]{Polynomial-time reductions between capacitated Steiner tree and vertex-disjoint paths problems.}
		\label{fig:pathsteiner}
	\end{center}
\end{figure}

The relations proved in this section are summarized in Figure \ref{fig:pathsteiner}, where a trivial polynomial reduction corresponds to a generalization of a special case. We recall that \texttt{VDISJ-PATH} is NP-complete in digraphs, even with $p=2$ souce-sink pairs \cite{fortune},  while it is polynomial-time solvable in undirected graph \cite{robertson} and DAGs \cite{fortune} when $p$ is fixed.

 We close this section by mentioning that the reductions given in Theorems  \ref{th:reductionFromMVDP} and \ref{cor:pMVDP-unif} are FPT-reductions  \cite{downey} with parameters $p$ and $K = O(p^2)$.

\section{NP-hardness of the general case}

In this section, we prove some NP-hardness and NP-completeness results. We first show that \texttt{CAP-STEINER-TREE} in digraphs is NP-complete even if $K\geq 3$ is fixed,  the minimum capacity $c_{\min}$  is any value in $\{1, \dots, K-2\}$, while the maximum capacity $c_{\max}$\ is at least 2.

\begin{theorem}\label{th:CSTP-directed-unif-K-fixed}
\texttt{CAP-STEINER-TREE} is NP-complete in digraphs, even if $K \geq 3$ is fixed, for any $c_{\min}\in \{1, \dots, K-2\}$ and $c_{\max} \geq 2$, with  $c_{\min} \leq  c_{\max}$.
\end{theorem}

\begin{proof}
This is a direct consequence of Theorem \ref{th:reductionFromMVDP} (for $K=3$, $c_{\min}=1$ and $c_{\max}=2$) and Theorem \ref{cor:2MVDP-unif} (for $K \geq 4$). Indeed, \texttt{VDISJ-PATH} is NP-complete in digraphs with $p=2$ source-sink pairs \cite{fortune}, and the two theorems show how to polynomially reduce \texttt{VDISJ-PATH} in this case to \texttt{CAP-STEINER-TREE} with the right number of terminals. Notice that we can fix the values of $c_{\min}$ and $c_{\max}$ in the constructed \texttt{CAP-STEINER-TREE} instances (with $c_{\min}\in \{1, \dots, K-2\}$, and $c_{\max} \geq 2$) by assigning these two values to two different arcs incident to terminals.
\end{proof}

\vspace{-0.3cm}\begin{coro}
 \texttt{CAP-STEINER-TREE} with uniform capacity $c\in\{2, \dots,K-2\}$ is NP-complete in digraphs, even if $K \geq 4$ is fixed.
\end{coro}

\noindent For undirected graphs, we have the following result:

\begin{theorem}\label{th:CSTP-undirected-non-unif}
\texttt{CAP-STEINER-TREE} is NP-complete in undirected graphs, even if $K \geq 3$ is fixed and if the minimum capacity $c_{min}$ and the maximum capacity $c_{\max}$ are two fixed constants, with $c_{\min} \in \{1, \dots, K-2\}$ and $c_{\min} < c_{\max}$.
\end{theorem}

\begin{proof}
We give a polynomial-time reduction from \texttt{SAT}. Assume first that $K=3$ and all edge capacities are equal to 1 or 2.
Let $X=\{x_1,...,x_{\xi}\}$ be the set of variables and let  $C=\{C_1,...,C_\nu\}$ be the set of clauses in an arbitrary instance $I$ of \texttt{SAT}. For each variable $x_i$, we denote by $o_i$ (resp. $\bar{o}_i$) the number of occurrences of $x_i$ (resp. $\bar{x}_i$) in the clauses. We can assume, without loss of generality, that $o_i \geq \bar{o}_i$ for every $i$ (by exchanging $x_i$ and $\bar{x_i}$ everywhere in the clauses, if necessary). The following instance $I'$ of \texttt{CAP-STEINER-TREE} is associated to $I$.

For each variable $x_i$, we construct a \emph{variable gadget} as follows: we add two vertices $v^i_0$ and $v^i_{2 o_i+1}$, and two vertex-disjoint paths between them. The first one, $\mu_i$, corresponding to literal $x_i$, is $v^i_0, v^i_1, \ldots, v^i_{2 o_i+1}$, where, for each $j \in \{1, \dots, o_i\}$, the edge $v^i_{2j-1}  v^i_{2j}$ has capacity 2, and, for each $j \in \{0, \dots, o_i\}$, the edge $v^i_{2j} v^i_{2j+1}$ has capacity 1. The second path, $\bar{\mu}_i$, corresponding to literal $\bar{x}_i$, is $v^i_0, \bar{v}^i_1, \ldots,\bar{v}^i_{2 \bar{o}_i}, v^i_{2 o_i+1}$, where, for each $j \in \{1, \dots, \bar{o}_i\}$, the edge $\bar{v}^i_{2j-1}  \bar{v}^i_{2j}$ has capacity 2, and, for each $j \in \{1, \dots, \bar{o}_i-1\}$, the edge $\bar{v}^i_{2j} \bar{v}^i_{2j+1}$ has capacity 1. The edges $v^i_0 \bar{v}^i_1$ and $\bar{v}^i_{2 \bar{o}_i} v^i_{2 o_i+1}$ also have capacity 1. The variable gadgets are linked together as follows: for each $i \in \{1, \dots, \xi-1\}$, there is an edge $v^i_{2 o_i+1} v^{i+1}_{0}$ of capacity 1.

For each clause $C_j$, we construct a \emph{clause gadget} as follows: we add two vertices $u^j_1, u^j_2$, and, for each literal $x_i$ (or $\bar{x}_i$) contained in $C_j$, we add edges $u^j_1 v^i_{2\ell-1}$ (or $u^j_1 \bar{v}^i_{2\ell-1}$) and $v^i_{2\ell} u^j_2$ (or $\bar{v}^i_{2\ell} u^j_2$) of capacity 2, if this literal occurs $\ell-1$ times in clauses $C_1, \dots, C_{j-1}$. The clause gadgets are linked together as follows: for each $j \in \{1, \dots, \nu-1\}$, there is an edge $u^j_2 u^{j+1}_1$ of capacity 2.

We complete the construction of $I'$ by adding a root $r$ and 3 terminals $t_1,t_2,t_3$, as well as 3 edges $v^{\xi}_{2 o_{\xi}+1} t_1$, $u^\nu_2 t_2$ and $u^\nu_2 t_3$ of arbitrary capacity (1 is fine, but any value fits), an edge $r v^1_0$ of capacity 1, and an edge $r u^1_1$ of capacity 2. The construction of $I'$ is illustrated in Figure \ref{fig:satcstp} for $I$ with $X=\{x_1,x_2,x_3\}$ and $C=\{x_1\bar{x}_2,x_1x_2x_3,\bar{x}_1x_2\bar{x}_3\}$. Solid lines have capacity 2 while dotted edges have capacity 1.

\begin{figure} [htbp]
	\begin{center}
		\includegraphics[scale=0.85]{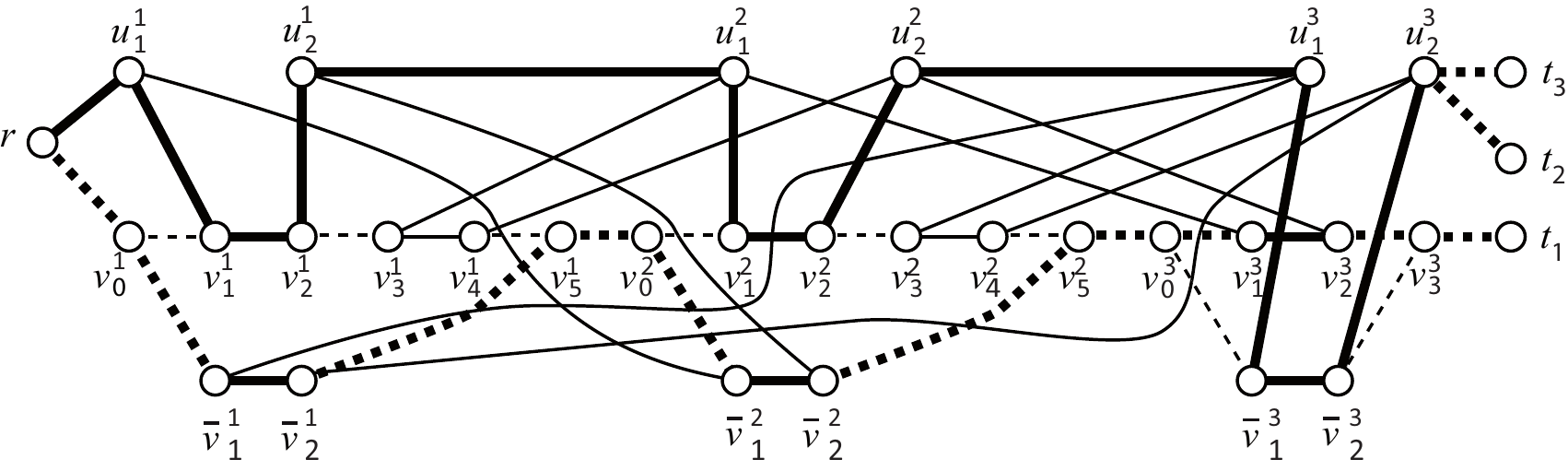}
		\caption{From \texttt{SAT} to  \texttt{CAP-STEINER-TREE}.}
		\label{fig:satcstp}
	\end{center}
\end{figure}

Let $S$ be a feasible solution to $I'$ (if any). To avoid cycles, the paths from $r$ to $t_2$ and $t_3$ must use the same sub-path $\pi_2$ from $r$ to $u_2^\nu$, and thus all edges of $\pi_2$ must have capacity 2.  So, $\pi_2$ starts with the edge $ru_1^1$. Then, the only possibility is to use the edge $u^1_1 v^i_1$ (or $u^1_1 \bar{v}^i_1$) for some $i$, and then the edges $v^i_1 v^i_2$ and $v^i_2 u^1_2$ (or $\bar{v}^i_1 \bar{v}^i_2$ and $\bar{v}^i_2 u^1_2$). The next step is to use the edge $u^1_2 u^2_1$. Using similar arguments with increasing values of $j$, we get that $\pi_2$ necessarily contains all edges $u^j_2 u^{j+1}_1$ with $1\le j <\nu$, and  ends at $u^\nu_2$ (which is adjacent to $t_2$ and $t_3$).

Since $\pi_2$ starts with the edge $ru_1^1$, the path $\pi_1$ from $r$ to $t_1$ starts with the edge $rv_0^1$. Moreover, to avoid cycles, $\pi_1$ and $\pi_2$ are internally vertex-disjoint. Since $u_1^j$ and $u_2^j$ belong to $\pi_2$ for all $j$, we conclude that, for each $i$, either  $\pi_1$ contains $\mu_i$ and then  $\pi_2$ may contain only edges of $\bar{\mu}_i$, or  $\pi_1$ contains $\bar{\mu}_i$ and then  $\pi_2$ may contain only edges of $\mu_i$.
This means that, for each $j$, there is a subpath of three edges of $\pi_2$, from $u_1^j$ to $u_2^j$, containing one edge of $\mu_i$ (resp. $\bar{\mu}_i$) for $i$ such that $x_i$ (resp. $\bar{x}_i$) is one of the literals contained in clause $C_j$, and there is no $k$ for which the subpath from $u_1^k$ to $u_2^k$ contains one edge of $\bar{\mu}_i$ (resp. $\mu_i$). We can therefore define a satisfying truth assignment  $\tau: X \rightarrow \{true,false\}$ as follows: for each $i$, if $\mu_i$ is a subpath of $\pi_1$, then  $\tau(x_i)=false$, else $\tau(x_i)=true$.

Conversely, if there is a satisfying truth assignment $\tau$ for $I$, we construct a feasible solution for $I'$ as follows. The path $\pi_1$ from $r$ to $t_1$ begins with the edge $rv_0^1$ and, for each $i$, $\pi_1$ has $\mu_i$ as a subpath if $\tau(x_i)=false$, and it has $\bar{\mu}_i$ as a subpath if $\tau(x_i)=true$. The path $\pi_2$ from $r$ to $u^\nu_2$ begins with the edge $ru_1^1$ and can be constructed sequentially by using edges not contained in $\pi_1$ (this is always possible, since $\tau$ is a satisfying truth assignment for $I$). The solution to $I'$ is then obtained by adding edges $v^{\xi}_{2 o_\xi +1} t_1$, $u^\nu_2 t_2$ and $u^\nu_2 t_3$ to $\pi_1\cup\pi_2$. The solution to $I'$ corresponding to the truth assignment $\tau(x_1,x_2,x_3)=(true,true,false)$ for $I$ is represented in Figure \ref{fig:satcstp} with bold lines.

In order to generalize this reduction to any $K \geq 3$, any $c_{\min} \in \{1, \dots, K-2\}$, and any $c_{\max}>c_{\min}$, we attach $c_{\min}+1$ terminals to $u_2^\nu$ (instead of 2) and $K-c_{\min}-2$ terminals to $r$ (instead of 0): the edges with capacity 1 and 2 become edges with capacity $c_{\min}$ and $c_{\min}+1$, respectively. Moreover, since the edges incident to the terminals can have any capacity, we can set the capacity of one of them to $c_{\max}$.
\end{proof}

Notice that the previous result is not valid for DAGs since the graphs constructed in the above proof possibly contain circuits. We now prove a complexity result for \texttt{LAB-VDISJ-PATH}, i.e. \texttt{ML-LAB-VDISJ-PATH} with lengths 0.

\begin{coro}\label{cor:hardness-MLVDP-undirected}
\texttt{LAB-VDISJ-PATH} with $p$ source-sink pairs is NP-complete in undirected graphs, and hence in digraphs, for any fixed $p \geq 2$, even if
$L_i$ contains all labels for every $i<p$ and $L_p$ contains all labels but one.
\end{coro}

\begin{proof}
For $p=2$, this follows from the proof of Theorem \ref{th:CSTP-undirected-non-unif}. Indeed, consider the two source-sink pairs $(s_1,s'_1)=(v_0^1,v^{\xi}_{2 o_\xi+1})$ and $(s_2,s'_2)=(u_1^1,u_2^\nu)$, identify the label of each edge with its capacity, and set $L_1=\{1,2\}$ and $L_2=\{2\}$. We have shown that there is a feasible solution to the instance $I$ of \texttt{SAT} if and only if there are two vertex-disjoints paths $P_1$ and $P_2$ linking $s_1$ to $s'_1$ and $s_2$ to $s'_2$, and such that $P_1$ uses edges of label 1 or 2, while $P_2$ uses only edges with label 2.
For larger values of $p$, we simply add dummy source-sink pairs.
\end{proof}

Note that this is in contrast with \texttt{VDISJ-PATH}, which is polynomial-time solvable in undirected graphs \cite{robertson} and in DAGs \cite{fortune} when $p$ is fixed. A similar problem has been studied in \cite{wu}. More precisely, \texttt{TWOCOL-VDISJ-PATH} is defined as follows: given an undirected graph in which every edge has color 1 or 2, and two source-sink pairs $(s_1, s'_1)$ and $(s_2, s'_2)$, determine whether $G$ contains two vertex-disjoint paths, the first one from $s_1$ to $s'_1$ using only edges of color 1, and the second one from $s_2$ to $s'_2$ using only edges of color 2.
\texttt{TWOCOL-VDISJ-PATH} is the special case of \texttt{ML-LAB-VDISJ-PATH} where there are $p=2$ source-sink pairs, all lengths are 0, and $L_i=\{i\}$ for each $i \in \{1,2\}$. It is proved in
 \cite{wu} that this problem is NP-complete.
The reduction given in the proof of Theorem \ref{th:CSTP-undirected-non-unif} yields an alternative proof of the NP-completeness of \texttt{TWOCOL-VDISJ-PATH}: indeed, as in the previous corollary, we can fix $(s_1,s'_1)=(v_0^1,v^{\xi}_{2 o_{\xi}+1})$ and $(s_2,s'_2)=(u_1^1,u_2^\nu)$, and identify the colors of the edges with their capacities. Also, for every variable gadget, we
add a path consisting of two edges of capacity (color) 1 between the endpoints of the edges of capacity 2 so that $P_1$ (the path that goes through edges of capacity 1 and 2) can avoid edges with color 2.\\

\noindent The next result deals with DAGs and undirected graphs with uniform capacity.

\begin{theorem}
\label{th:unif-DAG-SAT}
\texttt{CAP-STEINER-TREE} is NP-complete  in DAGs and undirected graphs, even in the case of uniform capacity $c$, for any $c \geq 2$ (not depending on $K$).
\end{theorem}

\begin{proof}
Let \texttt{3-SAT}$_3$ be the satisfiability problem in which every clause contains at most 3 variables, every variable appears in at most 3 clauses, and every litteral (a variable or its complement) appears in at most 2 clauses. \texttt{3-SAT}$_3$ is known to be NP-complete \cite{tovey}. We show how to polynomially reduce \texttt{3-SAT}$_3$ to \texttt{CAP-STEINER-TREE} with uniform capacity $c \geq 2$.

Let $I=(X,C)$ be an instance of \texttt{3-SAT}$_3$ with $X=\{x_1,...,x_{\xi}\}$ as set of variables and $C=\{C_1,...,C_\nu\}$ as set of clauses. To obtain an instance $I'$ of \texttt{CAP-STEINER-TREE} with uniform capacity $c$, we construct the following graph $G=(V,E)$: for every variable $x_i \in X$, we create three vertices $v_i,\bar{v}_i,s_i$ and $c$ terminals $TV_{i,1},\ldots, TV_{i,c}$; for every clause $C_j \in C$ we create a terminal $TC_j$; finally we add a vertex $r$. For the directed case, we consider the following arcs:  for every $i=1,...,{\xi}$, we create the arcs $(r,v_i),(r,\bar{v}_i),(v_i,s_i),(\bar{v}_i,s_i),(s_i,TV_{i,1}),\ldots, (s_i, TV_{i,c})$; for every $j=1,...,\nu$ we create the arcs $(v_i,TC_j)$ (resp. $(\bar{v}_i,TC_j)$) if $x_i$ (resp. $\bar{x}_i$) is in $C_j$. All the arcs have capacity $c$. The resulting graph $G$ can clearly be obtained in polynomial time and is a DAG. For the undirected case, we consider the same graph, but we replace each arc by an edge. The construction is illustrated in Figure \ref{fig:sat3} for $c=2$ and the \texttt{3-SAT}$_3$ instance where $X=\{x_1,x_2,x_3,x_4\}$ and $C=\{x_1 \bar{x}_2 x_3,\bar{x}_2 \bar{x}_3 x_4\}$, the arcs being oriented from $r$ down to the terminals in the directed case.

\begin{figure}[htbp]
\begin{center}
\includegraphics[scale=0.4]{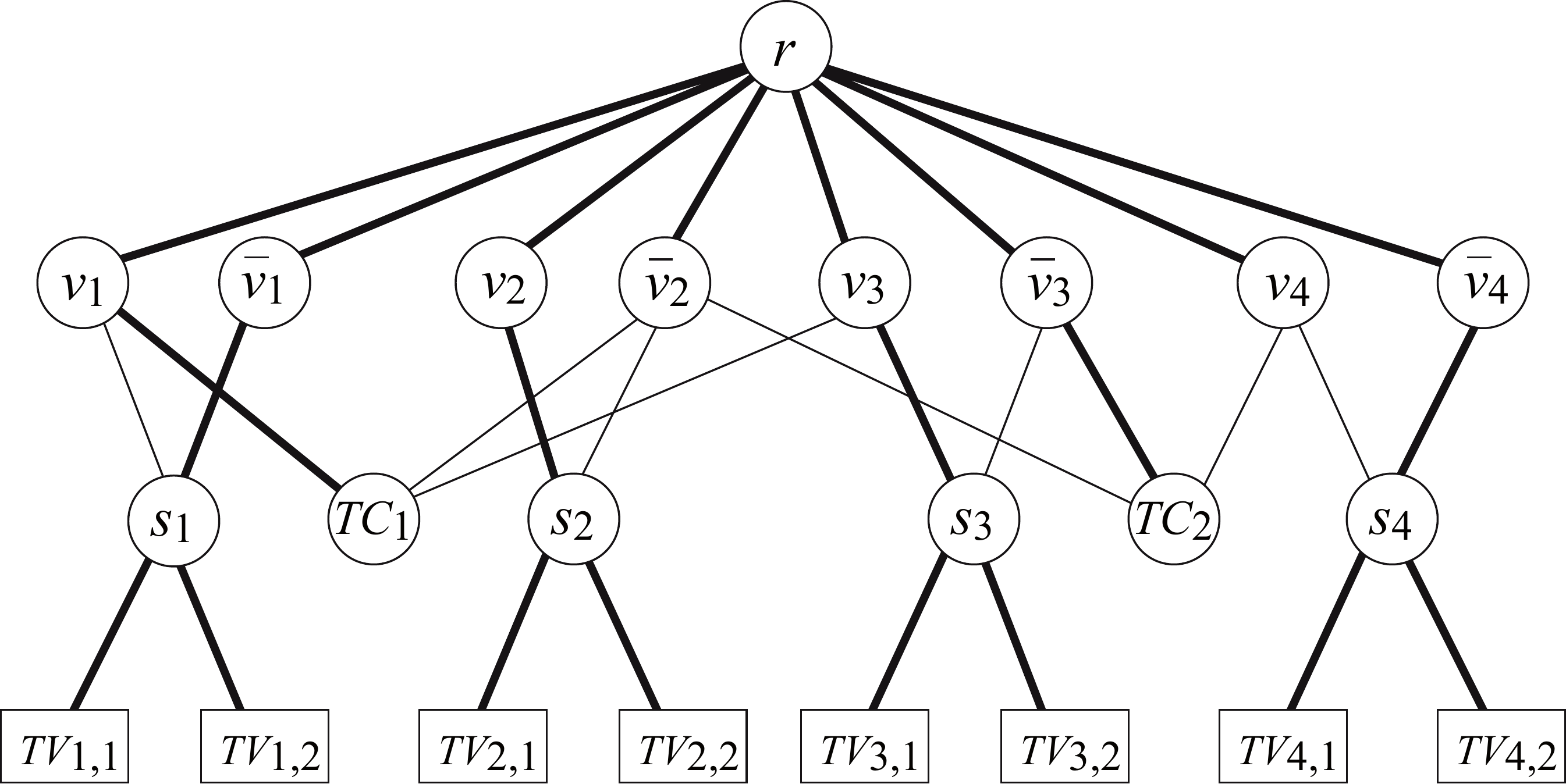} 
\caption{From  \texttt{3-SAT}$_3$ to \texttt{CAP-STEINER-TREE}.}
\label{fig:sat3}
\end{center}
\end{figure}

Assume there is a truth assignment $\tau: X \rightarrow \{true,false\}$ for $I$. We construct a feasible solution $S$ to $I'$ as follows. If $\tau (x_i)=false$ (resp. $\tau (x_i)=true$) then we include $(v_i,s_i)$ (resp.  $(\bar{v}_i,s_i)$) in $S$. For each clause $C_j$, we choose one of the true litterals in $C_j$, say $x_i$ (resp. $\bar{x_i}$), and add the arc $(v_i,TC_j)$ (resp. $(\bar{v}_i,TC_j)$) to $S$. Finally, for all $i=1,...,{\xi}$, we add the arcs $(r,v_i)$, $(r,\bar{v}_i)$, $(s_i,TV_{i,1}),\ldots,(s_i, TV_{i,c})$ to $S$. Clearly, $S$ is a tree rooted at $r$ and spanning all the terminals; in fact, $S$ is a spanning tree. For every $i=1,\ldots,\xi$ with $\tau (x_i)=false$ (resp. $\tau (x_i)=true$), $r$ is $(r,v_i)$-linked (resp. $(r,\bar{v}_i)$-linked) to $TV_{i,1},\ldots,TV_{i,c}$, and is
$(r,\bar{v}_i)$-linked (resp. $(r,v_i)$-linked) to at most two terminals $TC_j$ associated with clauses containing $\bar{x_i}$ (resp. $x_i$), since there are at most two such clauses. Since $c\geq 2$, all capacity constraints are satisfied. The solution $S$ associated with the truth assignment $\tau (x_1,x_2,x_3,x_4)=(true,false,false,true)$ is represented in Figure \ref{fig:sat3} with bold lines.

Now, let $S$ be a feasible solution to $I'$. Consider first the directed case. For all $j=1,...,\nu$, there is at least one index $i$ such that either  $(v_{i},TC_j) \in S$ and we then set $\tau (x_{i})=true$ or $(\bar{v}_{i},TC_j)\in S$ and we then set $\tau (x_{i})=false$. If a variable did not get any value, we arbitrarily choose one, say $true$. This gives a truth assignment satisfying each clause $C_j$. Let us verify that we have not assigned simultaneously values $true$ and $false$ to some variable. Notice first that, since $S$ has no cycle,
$r$ is either $(r,v_i)$-linked or $(r,\bar{v}_i)$-linked to the $c$ terminals $TV_{i,1},\ldots, TV_{i,c}$. Since all capacities equal $c$, this means that either $v_{i}$ or $\bar{v}_i$ has no $TC_j$ $(j=1,\ldots,\nu)$ as successor, which means that we do not assign both values $true$ and $false$ to a variable $x_i$.

Consider now the undirected case. If there is an index $i$ such that $(s_i,v_i) \in S$, then
$(\bar{v}_i,s_i) \in S$  since, in this case, $r$ must be $(\bar{v}_i,s_i)$-linked to $TV_{i,1},\ldots,TV_{i,c}$ in $S$. Hence, $(r,v_i) \notin S$ (otherwise, there would be a cycle in $S$) and we can replace
$(s_i,v_i)$ by $(r,v_i)$. Similarly, if $(s_i,\bar{v} _i) \in S$ for some index $i$, we replace this arc by $(r,\bar{v} _i)$. Assume now that $(v_i,TC_j) \in S$ and $(r,v_i) \notin S$ for some  $i\in\{1,\ldots,\xi\}$ and $j\in \{1,\ldots,\nu\}$. Then, we have $(TC_{h},v_i)\in S$ for some index $h\neq j$ and we replace $(TC_{h},v_i)$ by $(r,v_i)$. We make a similar exchange if $(\bar{v}_i,TC_j)\in S$ and $(r,\bar{v}_i) \notin S$. After having performed all these replacements, we get a new spanning tree $S'$, since each vertex (except the root) still has exactly one incoming arc.  Moreover, $S'$ has the same structure as the tree $S$ analyzed in the directed case. We can therefore obtain a satisfying truth assignment using the same rules as above.
\end{proof}

Remember (see Section 1) that \texttt{EDGE-COST-FLOW} consists in determining a\break minimum-length feasible flow of $K$ units from $s$ to $t$ in a given graph, where the length of a flow is the total length of the arcs/edges carrying a positive flow. \texttt{EDGE-COST-FLOW} is very close to \texttt{ML-CAP-STEINER-TREE}. Indeed, given an instance of \texttt{ML-CAP-STEINER-TREE} in a graph $G$ with $K$ terminals, we can construct a graph $G'$ obtained from $G$ by adding a vertex $r'$ and linking every terminal to this new vertex. Solving \texttt{EDGE-COST-FLOW} in $G'$ with $s=r$ and $t=r'$ is then equivalent to solving \texttt{ML-CAP-STEINER-TREE} in $G$, except that a feasible solution is not required to be a tree. The proof of Theorem \ref{th:unif-DAG-SAT} provides the following corollary:

\begin{coro}
\texttt{EDGE-COST-FLOW} is NP-hard in DAGs and undirected graphs, even if all lengths are 1 and all capacities are 1 or 2.
\end{coro}

\begin{proof}
 Consider first the directed acyclic case. Given an instance $I=(X,C)$ of \texttt{3-SAT}$_3$, let us construct a graph $G=(V,E)$, as in the proof of Theorem \ref{th:unif-DAG-SAT}, setting $c=2$. We then add a vertex $r'$ and link $TV_{i,1},TV_{i,2}$ $(i=1,\ldots,\xi)$ and  $TC_j$ $(j=1,\ldots,\nu)$ to $r'$. For all arcs $e$ incident to $r$, as well as those linking $s_i$ to $TV_{i,1}$ and $TV_{i,2}$, we set $\ell(e)=1$, while $\ell(e)$ is set to $4\xi+1$ for all other arcs $e$. All arcs have capacity 2, except those incident to $r'$ which have capacity 1. We then solve an \texttt{EDGE-COST-FLOW} instance, looking for a flow of $K=2 \xi +\nu$ units from $r$ to $r'$. Let $I'$ be this instance. We prove that $I$ is satisfiable if and only if the total length of an optimal solution to $I'$ is at most $L=(4\xi+1)(3\xi+2\nu)+4\xi$.

 If there is a satisfying truth assignment  for $I$, we construct a solution $S$ to $I'$ as in the proof of Theorem \ref{th:unif-DAG-SAT}, except that we add an arc from every terminal to $r'$. It is not difficult to check that the total length of such a solution is at most $L$.

 Consider now a solution $S$ to $I'$ of total length at most $L$. Since the $2 \xi +\nu$ arcs incident to $r'$ have capacity 1, they all carry a positive flow.
 Hence, for every $j=1,\ldots,\nu$, there is at least one index $i$ such that $(v_i,TC_j)$ or $(\bar{v}_i,TC_j)$ carries a positive flow in $S$. Moreover, for every $i=1,\ldots,\xi$, at least one of the arcs $(v_i,s_i)$ and $(\bar{v}_i,s_i)$ carries a positive flow in $S$. Therefore, the total length of $S$ is at least $(4\xi+1)(3 \xi+2\nu)= L-4\xi$, which means that no other arc $e$ with $\ell(e)=4\xi+1$ can belong to $S$. In particular, for every $j=1,\ldots,\nu$, there is {\em exactly} one arc in $S$ with a positive flow linking a vertex in $\{v_1,\bar{v}_2, v_2,\bar{v}_2,\ldots,v_{\xi},\bar{v}_{\xi}\}$ to $TC_j$, and, for every $i=1,...,\xi$, {\em exactly} one of the arcs $(v_i,s_i)$ and $(\bar{v}_i,s_i)$ carries a positive flow flow in $S$. If $(v_i,s_i)$ (resp. $(\bar{v}_i,s_i)$) carries a positive flow in $S$, then there is no flow on the arcs linking $v_i$ (resp. $\bar{v}_i$) to a $TC_j$ since $(r,v_i$) (resp. $(r,\bar{v}_i)$) has capacity 2 while 2 units of flow are used to reach $TV_{i,1}$ and $TV_{i,2}$. Hence the structure of $S$ is the same as in the proof of
 Theorem \ref{th:unif-DAG-SAT}, and we can set $x_i=false$ if $(v_i,s_i)$ carries a positive flow in $S$, and $x_i=true$ otherwise, to obtain a satisfying truth assignment for $I$.

 To obtain a graph with uniform length 1, we replace each arc $e$ by a path with $\ell(e)$ arcs of length 1.

 The proof is similar for the undirected case.
\end{proof}

\noindent We close this section by considering a final uniform case for \texttt{ML-CAP-STEINER-TREE}.

\begin{theorem}
\label{th:unit-cost-Steiner}
 \texttt{ML-CAP-STEINER-TREE} is NP-hard in any graph, even if all capacities are equal to $K-\kappa$ for any positive constant $\kappa$ and if all lengths are 1.
\end{theorem}

\begin{proof}
It is well-known that \texttt{STEINER-TREE} is NP-hard even in the case of unit lengths \cite{garey}. Given any positive number $\kappa$, we show how to polynomially reduce \texttt{STEINER-TREE} with unit lengths to \texttt{ML-CAP-STEINER-TREE} with uniform capacity $K-\kappa$ and unit lengths. Let $I'$ be an instance of \texttt{STEINER-TREE} in a graph $G'=(V',E')$ with $K'$ terminals and $\ell(e)=1$ for all $e\in E$. We construct a graph $G=(V,E)$ from $G'$ by adding $\kappa$ terminals and linking them to $r$ with edges/arcs of length 1. We then set all capacities to $K'$ to obtain an instance $I$ of \texttt{ML-CAP-STEINER-TREE} with $K=K'+\kappa$ terminals, $\ell(e)=1$ and $c(e)=K-\kappa$ for all $e\in E$. Clearly, solving $I'$ is equivalent to solving $I$, and $I$ can be built from $I'$ in polynomial time.
\end{proof}

\section{{\normalsize ML-CAP-STEINER-TREE} with a fixed number of terminals}

In this section, we assume that the number $K$ of terminals is fixed. The first theorem deals with undirected graphs having uniform capacity, and complements Theorems \ref{th:CSTP-undirected-non-unif} and \ref{th:unif-DAG-SAT}.

\begin{theorem}\label{th:undirected-unif-fixedK}
In undirected graphs with a fixed number $K$ of terminals and uniform capacity, \texttt{CAP-STEINER-TREE} is solvable in polynomial time, and  \texttt{ML-CAP-STEINER-TREE} is polynomially equivalent to \texttt{ML-VDISJ-PATH} with a fixed number of source-sink pairs.
\end{theorem}
\begin{proof}
The first part of the theorem is a direct consequence of Theorem \ref{cor:MCSTPuniftoMVDP}, since \texttt{VDISJ-PATH} is polynomial-time solvable --and even FPT-- in undirected graphs when the number of source-sink pairs is fixed \cite{robertson}. The second part comes from Theorems \ref{cor:pMVDP-unif} and \ref{cor:MCSTPuniftoMVDP}.
\end{proof}

On the one hand, recall that, for a fixed number $p$ of source-sink pairs, the complexity of  \texttt{ML-VDISJ-PATH} in undirected graphs is open for a long time \cite{kobayashi} (and so determining the one of \texttt{ML-CAP-STEINER-TREE} in this case is as hard as settling this open problem); however, there exists a probabilistic polynomial-time algorithm to solve the case with two source-sink pairs \cite{bjorklund}, although no deterministic one is known yet. On the other hand, Corollary \ref{cor:hardness-MLVDP-undirected} shows that  \texttt{ML-LAB-VDISJ-PATH} with lengths 0 is already NP-complete in undirected graphs when $p \geq 2$ is fixed. The next theorem shows that \texttt{ML-LAB-VDISJ-PATH} is tractable in DAGs if $p$ is fixed, which will come in handy for proving that \texttt{ML-CAP-STEINER-TREE} is solvable in polynomial time in DAGs if $K$ is fixed.

\begin{theorem}
\label{th:CDS-DAG-K-fixed}
 In DAGs, \texttt{ML-LAB-VDISJ-PATH} is solvable in polynomial time for any fixed number of source-sink pairs.
\end{theorem}
\begin{proof}

We solve \texttt{ML-LAB-VDISJ-PATH} by using a dynamic programming algorithm. More precisely, consider a DAG $G=(V,E)$ with a label $\lambda(e)$ on each arc $e\in E$, and with $p$ vertex-disjoint pairs $(s_i, s'_i)$ and their associated label sets $L_i$. We first order the vertices of $G$ using a topological ordering, and denote by $num(v)$ the position of each vertex $v$ in such an ordering (i.e., $num(u) < num(v)$ for all $(u,v)\in E$).

Let $\mathcal{P}$ be the set of $p$-tuples $(v_1, \dots, v_p)$ of vertices of $G$, and let $f:\mathcal{P} \rightarrow \mathbb{N}$ be the function such that $f(v_1, \dots, v_p)$ is the minimum total length of a set of $p$ vertex-disjoint paths in $G$ such that the $i$th one goes from $s_i$ to $v_i$ and uses only arcs with labels in $L_i$.
If $num(v_i) \leq num(s_i)$ for all $i=1,\ldots,p$, then $f(v_1,\dots,v_p)=0$ if $v_i=s_i$ for all $i$, and
$f(v_1,\dots,v_p)=+\infty$ otherwise.
Consider now a $p$-tuple $(v_1, \dots, v_p)$ such that $num(v_i) > num(s_i)$ for at least one index $i$, and let $h$ be the index such that $num(v_h)=\max_{i:num(v_i)>num(s_i)}\{num(v_i)\}$.
If $v_h=v_i$ for some $i \neq h$, then $f(v_1, \dots, v_p)=+\infty$. Otherwise, let $\mathcal{F}$ be the set of vertices $v$ such that $num(v)\geq num(s_h)$ and there exists an arc $(v,v_h)$ whose label is in $L_h$. If $\mathcal{F}=\emptyset$ then
$f(v_1, \dots, v_p)=+\infty$; otherwise, we have
$$f(v_1, \dots, v_{h-1}, v_h, v_{h+1}, \dots, v_p)=\min_{v \in \mathcal{F}} \{\ell(v,v_h)+f(v_1, \dots, v_{h-1}, v, v_{h+1}, \dots, v_p)\}.$$
The dynamic programming algorithm works as follows. For $val$ from 2 to $n$, we enumerate all $p$-tuples $v_1, \dots, v_p$ with $\max_{i:num(v_i)>num(s_i)}\{num(v_i)\}=val$ and, for each of them, we compute the corresponding value of $f$. The number of enumerated $p$-tuples is thus in $O(n^{p+1})$ and the value of each one can be computed in $O(n)$, which yields an $O(n^{p+2})$-time algorithm. The optimal {\color{black}value} to the \texttt{ML-LAB-VDISJ-PATH} instance is then equal to $f(s'_1,\dots,s'_p)$.
\end{proof}

\begin{theorem}
\label{th:DAG-fixedK}
\texttt{ML-CAP-STEINER-TREE} is solvable in polynomial time in DAGs if $K$ is fixed.
\end{theorem}
\begin{proof}
This is a direct consequence of Theorems \ref{th:reductionToMLVDP} and \ref{th:CDS-DAG-K-fixed}.
\end{proof}

Notice that \texttt{ML-LAB-VDISJ-PATH} in DAGs cannot be FPT in $p$ (unless FPT=W[1]), since \texttt{VDISJ-PATH} (i.e., the special case where all arcs have zero length and $L_i=\{1,\dots,k\}$ for each $i$) is W[1]-hard with respect to $p$ in DAGs \cite{slivkins}. Moreover, \texttt{ML-CAP-STEINER-TREE} cannot be FPT in $K$ (unless FPT=W[1]), since Theorem \ref{cor:pMVDP-unif} shows that \texttt{ML-VDISJ-PATH} can be FPT-reduced to \texttt{ML-CAP-STEINER-TREE} (with respective parameters $p$ and $K=p^2$).

\section{{\normalsize ML-CAP-STEINER-TREE} with large capacities}

In this section, we study the case where all capacities are almost equal to the number of terminals.
We first consider the case where the minimum capacity $c_{\min}$ is at least equal to $K-\kappa$, where $\kappa \geq 0$ is an arbitrary constant. In what follows, we denote by  $\rho$  the best possible approximation ratio for \texttt{STEINER-TREE} ($\rho \leq 1.39$ in undirected graphs \cite{byrka}), and by $\rho'$ the best possible approximation ratio for \texttt{ML-VDISJ-PATH} with a fixed number of source-sink pairs. As mentioned in the previous section, $\rho'=1$ in DAGs, and determining whether $\rho'=1$ or not in undirected graphs is a long-standing open problem.

Without loss of generality, we assume in this section that $\ell(e) \in \mathbb{N}^*$ for all $e \in E$. If this is not the case, we modify the lengths as follows: for all $e \in E$, we multiply $\ell(e)$ by $D|E|$ if $\ell(e)>0$, where $D$ is the lowest common multiple of the denominators of the lengths $\ell(e)$, and we set $\ell(e)=1$ if $e$ has length zero.

\subsection{{\normalsize ML-CAP-STEINER-TREE} with $c_{\min} \geq K-\kappa$ for any constant $\kappa \geq 0$}

The first result obtained in this section complements  Theorem \ref{th:unit-cost-Steiner} and generalizes the first part of Theorem \ref{th:undirected-unif-fixedK}. {\color{black}Notice that, from Theorem \ref{th:CSTP-directed-unif-K-fixed}, \texttt{CAP-STEINER-TREE} is NP-complete in digraphs with uniform capacity $c=K-\kappa$ for any constant $\kappa\geq 2$.}

\begin{theorem}\label{th:CSTP-large-capa-unif}
{\color{black}In DAGs and undirected graphs having uniform capacity $c=K-\kappa$, \texttt{CAP-STEINER-TREE} is solvable in polynomial time and \texttt{ML-CAP-STEINER-TREE} can be approximated within a ratio of $\rho'+\rho$, for any constant $\kappa\geq 0$.}
\end{theorem}

\begin{proof}
We  first state and prove some useful properties.
Let $I$ be an instance of \texttt{CAP-STEINER-TREE}.

\begin{claim}
\label{fact7.1}
There is a feasible solution $S$ for $I$ if and only if  there is a tree $S^R$ (called \emph{reduced} tree) rooted at $r$, spanning a subset $T'\subseteq T$ of terminals, and such that, for every edge $e$ incident to $r$ in $S^R$, $r$ is $\bar{e}$-linked to at least $\kappa$ terminals in $S^R$.
\end{claim}

\begin{proof}
A feasible solution for $I$ is clearly such a tree. So, assume that such a tree $S^R$ exists. We iteratively complete $S^R$ to obtain a Steiner tree $S$ spanning also the terminals in $T \setminus T'$ as follows. For every terminal $t \notin T'$, we consider any path $\mu$ from $r$ to $t$ in $G$, and we add to $S^R$ the subpath of $\mu$ from $v$ to $t$, where $v$ is  the vertex in $S^R \cap \mu$ the closest to $t$ on $\mu$. From the hypothesis, given any edge $e$ of $G$ (including those not in $S^R$), we know that $r$ is $\bar{e}$-linked to at least $\kappa$ terminals in $S^R$. Hence, $S$ does not violate the capacity constraints and is therefore a feasible solution to $I$.
\end{proof}

\begin{claim}
	\label{fact7.2}
If $S^R$ is a \emph{minimal}, i.e. inclusion-wise minimal, reduced tree, then each of its subtrees rooted at a vertex distinct from $r$ contains at most $\kappa$ terminals.
\end{claim}

\begin{proof}
Let $S^R$ be a minimal reduced tree, and assume it contains a vertex $v\neq r$ such that $S^R(v)$ contains at least $\kappa+1$ terminals. {\color{black}We can assume without loss of generality that $v$ is a child of $r$}. Since all terminals are leaves, we can remove one of {\color{black}the terminals of $S^R(v)$} from $S^R$ to obtain a smaller reduced tree, a contradiction.
\end{proof}

\begin{claim}
	\label{fact7.3}
A minimal reduced tree contains at most $2\kappa$ terminals.
\end{claim}

\begin{proof}
Assume a minimal reduced tree $S^R$ contains at least $2\kappa+1$ terminals, and consider any child $v$ of $r$ in $S^R$. It follows from the previous claim that $r$ is $\overline{(r,v)}$-linked to at least $\kappa+1$ terminals in $S^R$. {\color{black}All terminals being leaves}, we can therefore delete any terminal {\color{black}from $S^R$} to obtain a smaller reduced tree, a contradiction.
\end{proof}

We can now prove Theorem \ref{th:CSTP-large-capa-unif}. We first consider \texttt{CAP-STEINER-TREE}. According to Claim \ref{fact7.3} and Property \ref{prop:numbvert}, the undirected skeleton of a minimal reduced tree can contain up to $4\kappa$ vertices. We therefore enumerate all labelled trees (including potential skeletons of reduced trees) on at most $4\kappa$ vertices. We then orient each of them from the root to the leaves, and for each such rooted tree we try to replace the arcs by vertex-disjoint paths in a similar way as in Theorem \ref{cor:MCSTPuniftoMVDP}. If such a replacement is possible, we test whether the extended skeleton is a reduced tree: in such a case, we stop the enumeration since we know from Claim \ref{fact7.1} that the \texttt{CAP-STEINER-TREE} instance $I$ has a feasible solution. If no potential skeleton can be extended to a reduced tree, then $I$ has no solution: indeed, if such a solution $S$ exists, it contains a minimal reduced tree, whose skeleton is necessarily considered in our enumeration and then extended to a reduced tree, which leads to a contradiction.

{\color{black}There are $O(K^{2\kappa})$ ways of choosing at most $2 \kappa$ terminals among $K$. For each such choice of at most $2 \kappa$ terminals}, it then follows from Property \ref{enumskelet} that the set of potential skeletons of {\color{black}minimal reduced trees spanning these terminals} can be enumerated in {\color{black}$O(n^{2\kappa-1})$ time} (since $\kappa$ is a constant). {\color{black}Finally}, at most $4\kappa-1$ arcs must be replaced by vertex-disjoint paths in every potential skeleton (recall that this can be done in polynomial time in DAGs and undirected graphs, since $\kappa$ is a constant). Hence, the whole process takes a polynomial time.

Consider now a \texttt{ML-CAP-STEINER-TREE} instance $I'$. We proceed as above, but instead of choosing arbitrary vertex-disjoint paths, we solve the associated \texttt{ML-VDISJ-PATH} instance with a $\rho'$-approximation algorithm, and, instead of stopping the enumeration when a reduced tree is found, we enumerate all of them and store the best one, denoted by $S^1$.
Notice that the total length of $S^1$ is at most $\rho'$ times larger than the total length of an optimal solution to $I'$ since such an optimal solution contains a reduced tree. We then use a $\rho$-approximation algorithm to determine a minimum-length Steiner tree $S^2$ spanning all the terminals not already spanned by $S^1$. Clearly, the total length of $S^2$ is at most $\rho$ times larger than the total length of an optimal solution to $I'$.
We finally build a solution $S$ to $I'$ by removing from $S^1 \cup S^2$ all arcs of $S^2$ entering a vertex with in-degree 2 in $S^1 \cup S^2$. This yields  a $(\rho'+\rho)$-approximation algorithm to \texttt{ML-CAP-STEINER-TREE} (with $\rho'=1$ for DAGs).
\end{proof}

\vspace{-0.3cm}It follows from Theorem \ref{th:CSTP-undirected-non-unif} that \texttt{CAP-STEINER-TREE} with non-uniform capacities is NP-complete in undirected graphs {\color{black}when the minimum capacity $c_{\min}$ equals $K-\kappa$}, for any constant $\kappa \geq 2$. We show however that Theorem \ref{th:CSTP-large-capa-unif} can be extended to DAGs with non-uniform capacities.

\begin{theorem}\label{th:CSTP-large-capa-non-unif-DAGs}
In DAGs with $c_{\min}\geq K-\kappa$, \texttt{CAP-STEINER-TREE} is solvable in polynomial time and \texttt{ML-CAP-STEINER-TREE} can be approximated within a ratio of $1+\rho$, for any constant $\kappa \geq 0$.
\end{theorem}
\begin{proof}
As was the case for the previous theorem, we start with some claims. In particular, we extend the definition of a reduced tree to take into account the non-necessarily uniform capacities. {\color{black}Let $I$ be an instance of \texttt{CAP-STEINER-TREE} in a DAG with $c_{\min} \geq K-\kappa$ for some constant $\kappa\geq 0$. If $K < \kappa$, then the result follows from Theorem \ref{th:DAG-fixedK}. So, assume $K \geq \kappa$.}

\vspace{-0.3cm}\begin{claim}
\label{fact:reducedtree}
There is a feasible solution $S$ to $I$ if and only if there is a tree $S^R$ (called a \emph{reduced} tree) rooted at $r$, spanning at least $\kappa$ terminals, and such that, for each arc $a$ with capacity $c(a)$ (not only those incident to $r$), $r$ is $\bar{a}$-linked to at least $K-c(a)$ terminals in $S^R$.
\end{claim}

\vspace{-0.3cm}\begin{proof}
A solution for $I$ is clearly such a tree. Now, assume the existence of a reduced tree $S^R$ that spans a subset $T'$ of {\color{black} at least $\kappa$} terminals. We complete $S^R$ to obtain a Steiner tree $S$ spanning also the terminals in $T \setminus T'$, as in Claim \ref{fact7.1}. From the hypothesis, for each arc $a$ of $S^R$, $r$ is $\bar{a}$-linked to at least $K-c(a)$ terminals in $S^R$, and, for each arc $b$ of $G$ not in $S^R$, $r$ is $\bar{b}$-linked to at least $\vert T' \vert \geq \kappa\ge K-c(a)$ terminals of $S^R$. Hence, $S$ does not violate the capacity constraints and is thus a feasible solution to $I$.
\end{proof}

\vspace{-0.3cm}\begin{claim}
\label{fact:dagnonunif}
In any \emph{minimal}, i.e. inclusion-wise minimal, reduced tree $S^R$, no vertex has out-degree greater than $\kappa+1$.
\end{claim}

\vspace{-0.3cm}\begin{proof}
Assume that a minimal reduced tree $S^R$ contains a vertex $v$ with at least $\kappa+2$ outgoing arcs and let $v'$ be a child of $v$ in $S^R$.  Let $S'^R$ denote the subtree obtained from $S^R$ by removing all vertices of $S^R(v')$.
Since every subtree $S^R(w)$ of $S^R$ rooted at a child $w$ of $v$ contains at least one terminal (otherwise $S^R$ is not minimal), we know that $S'^R(v)$ (and thus also $S'^R$) spans at least $\kappa+1>\kappa$ terminals. Hence, for every arc $a$ of $S'^R$ not on the path from $r$ to $v$ in $S^R$ and not in $S'^R(v)$, we know that
$r$ is $\bar{a}$-linked to at least $\kappa+1>K-c(a)$ terminals in $S'^R$. For every child $w$ of $v$ in $S'^R$, we know that, {\color{black}for any arc $a$ in $S'^R(w)\cup\{(v,w)\}$, $r$ is $\bar{a}$-linked} to the at least $\kappa\geq K-c(a)$ terminals in the subtree of $S'^R(v)$ obtained by removing all the vertices of $S'^R(w)$. Finally, for any arc $a$ on the path from $r$ to $v$ in $S'^R$, we know that $r$ is $\bar{a}$-linked to at least  $K-c(a)$ terminals in  $S'^R$, since this was the case in $S^R$. {\color{black}Hence, we have proved that $S'^R$ satisfies the definition of a reduced tree, and is included in $S^R$ while being smaller, a contradiction.}
\end{proof}

\vspace{-0.3cm}\begin{claim}
	\label{fact:height}
Any directed path in the skeleton of a minimal reduced tree $S^R$ contains at most $\kappa+1$ arcs.
\end{claim}

\vspace{-0.3cm}\begin{proof}
Assume that the skeleton of a minimal reduced tree $S^R$ contains a directed path with at least $\kappa+2$ arcs. {\color{black}Without loss of generality}, we can choose  such a path $\mu$ from $r$ to a terminal $t\in T$, since each leaf is a terminal (otherwise $S^R$ is not minimal). We denote by $v$ the predecessor of $t$ in $\mu$. Since each internal vertex of $\mu$ has degree at least 3 in the skeleton of $S^R$ {\color{black}(and hence in $S^R$)}, {\color{black}and since $S^R$ is minimal}, we know that $S^R$ spans at least $\kappa + 2$ terminals. We now remove the path from $v$ to $t$ in $S^R$ and thus obtain a subtree $S'^R$ spanning at least $\kappa +1$ terminals. Using arguments similar to those in the proof of Claim \ref{fact:dagnonunif}, it is easy to check that $S'^R$ {\color{black}satisfies the definition of a reduced tree}, which means that $S^R$ was not minimal, a contradiction.
\end{proof}

We can now prove Theorem \ref{th:CSTP-large-capa-non-unif-DAGs}. It follows from Claims \ref{fact:dagnonunif} and \ref{fact:height} that the skeleton of a minimal reduced tree has maximum out-degree $\kappa+1$ and maximum height $\kappa+1$. Hence,
the undirected skeleton of a minimal reduced tree contains
 at most {\color{black}$\Lambda = (\kappa+1)^{\kappa+1}$ terminals (which are its leaves)}, and it follows from Property \ref{prop:numbvert} that such a skeleton has at most $2 \Lambda$ vertices.

We therefore enumerate all trees with at most $2 \Lambda$ vertices, keeping only those that span at least $\kappa$ and at most $\Lambda$ terminals, in order to ensure that any potential skeleton of a minimum reduced tree is enumerated. Each such enumerated tree is oriented from $r$ to the leaves, and we then try to replace its arcs by vertex-disjoint paths but, unlike in Theorem  \ref{th:CSTP-large-capa-unif}, when replacing an arc $(u,v)$ of a tree by a path, we impose that each arc of the path from $u$ to $v$ has a label (capacity) $\geq K-x$, where $x$ is the number of terminals  which are not descendant of $v$ in the tree. In other words, instead of solving a \texttt{VDISJ-PATH} instance, we solve a  \texttt{LAB-VDISJ-PATH} instance.
If all arcs can be replaced by labelled vertex-disjoint paths, we test whether the extended skeleton is a reduced tree: in such a case, we stop the enumeration, since we know from Claim \ref{fact:reducedtree} that $I$ has a feasible solution. Otherwise, we conclude that $I$ has no solution.

{\color{black}Let us show that the whole process takes a polynomial time. $\Lambda$ being a constant, this means that, from Cayley's formula and from the number of ways for choosing at most $2 \Lambda$ vertices among $n$, there is a polynomial number of  
labelled trees to enumerate. Besides, since the number of arcs that must be replaced by vertex-disjoint paths is at most  $2\Lambda-1$ in each tree, this means, from Theorem \ref{th:CDS-DAG-K-fixed}, that the associated \texttt{LAB-VDISJ-PATH} instance can be solved in polynomial time.}

Consider now a \texttt{ML-CAP-STEINER-TREE} instance $I'$. We proceed in a similar way as in Theorem \ref{th:CSTP-large-capa-unif}. More precisely, instead of choosing arbitrary vertex-disjoint paths, we solve an  \texttt{ML-LAB-VDISJ-PATH} instance with a fixed number of vertex pairs, which takes a polynomial time in DAGs according to Theorem \ref{th:CDS-DAG-K-fixed}.
However, instead of stopping the enumeration when a reduced tree is found, we enumerate all of them and store the best one, that we denote by $S^1$.
The total length of $S^1$ is a lower bound on the total length of an optimal solution to $I'$, since such an optimal solution contains a reduced tree. We then use a $\rho$-approximation algorithm to determine a minimum-length Steiner tree $S^2$ spanning all the terminals not already spanned by $S^1$. The total length of $S^2$ is at most $\rho$ times larger than the total length of an optimal solution to $I'$. We finally build a solution $S$ to $I'$ by removing from $S^1 \cup S^2$ all arcs of $S^2$ entering a vertex with in-degree 2 in $S^1 \cup S^2$. This yields a $(1+\rho)$-approximation algorithm for \texttt{ML-CAP-STEINER-TREE} in DAGs.
\end{proof}

\vspace{-0.5cm}\subsection{\texttt{ML-CAP-STEINER-TREE}  with $c_{\min} \geq K-1$}

The results given in this section generalize the main results known about the complexity and approximation of \texttt{STEINER-TREE}.

If $c_{\min} \geq K$, then any Steiner tree is a feasible solution to \texttt{CAP-STEINER-TREE}, and thus \texttt{ML-CAP-STEINER-TREE} is equivalent to \texttt{STEINER-TREE}, which can be solved in polynomial time when $K$ is fixed \cite{dreyfus,feldman,watel}. So consider the case where $c_{\min} = K-1$. In what follows, we denote by $E_K$ the subset of arcs/edges with capacity at least $K$.
 Let $I$ be an \texttt{ML-CAP-STEINER-TREE}  instance and let $S$ be an {\color{black}optimal} solution to $I$. Let $w$ be the  closest vertex  to $r$ in $S$ having out-degree at least 2 (with possibly $r=w$). All arcs on the path linking $r$ to $w$ are in $E_K$, while those in $S(w)$ can have any capacity since $r$ is $e$-linked to at most $K-1$ terminals for all $e$ in $S(w)$. Moreover, $S(w)$ spans all terminals and, since we can assume that all vertices without outgoing arcs in $S$ are terminals, we know that $S$ contains at most $K-1$ vertices of degree at least 3 (see the proof of Property \ref{prop:numbvert}).

Assume we can find in $G$ a vertex $w$ and two terminals $t_i$ and $t_j$ such that there are three internally  vertex-disjoint paths: $\mu_{rw}$ from $r$ to $w$ (with possibly $r=w$) with all its arcs in $E_K$, $\mu_{wt_i}$ from $w$ to $t_i$, and $\mu_{wt_j}$ from $w$ to $t_j$. We can then build a feasible solution to $I$ by extending $\mu_{rw} \cup \mu_{wt_i} \cup \mu_{wt_j}$ arbitrarily to obtain a Steiner tree spanning all terminals. Indeed, any arc in $\mu_{rw} \cup \mu_{wt_i} \cup \mu_{wt_j}$ has a residual capacity $\geq K-2$, and there are $K-2$ other terminals to span in order to get a Steiner tree. Conversely, if there is a feasible solution to $I$, then there is such a triple $(w,t_i,t_j)$.

 \begin{claim}
 \label{claim:shortpath}
 Let $S$ be an optimal solution to an instance $I$ of \texttt{ML-CAP-STEINER-TREE} in a graph $G=(V,E)$ with $c_{\min}=K-1$ and $\ell(e)>0$ for all $e\in E$. Let $w$ be the vertex with out-degree at least 2 the closest to $r$ in $S$, and let $\mu_{rw}$ be the path from $r$ to $w$ in $S$. Then all shortest paths from $r$ to $w $ in $G'=(V,E_K)$ intersect $S(w)$ only at $w$, and $\mu_{rw}$ is one of them.
 \end{claim}

 \begin{proof}
First notice that $S= \mu_{rw}\cup S(w)$ and all arcs of $\mu_{rw}$ belong to $E_K$. Let $\mu_{rw}'$ be any shortest path from $r$ to $w$ in $G'$, and let $W$ be the set of vertices that belong to both $\mu_{rw}'$ and $S(w)$. If $W = \{w\} $ then $S'= \mu_{rw}' \cup S(w)$ is a feasible solution to $I$, which means that  $\ell(\mu_{rw}) = \ell(\mu_{rw}')$, otherwise $S$ would not be optimal.

So assume $W \neq  \{w\} $ and let $S'$ be the tree obtained from $S$ by replacing  $\mu_{rw}$ by $\mu_{rw}'$, and by removing all arcs $(u,v)\notin \mu_{rw}'$ with $v \in W$, to ensure that each vertex (except $r$) still has in-degree 1. Notice that $\ell(S')<\ell(S)$, since
$\mu_{rw}'\leq \mu_{rw}$ and at least one arc (of length $>0$) is removed from $S(w)$ to obtain $S'$. Then, we remove all arcs  $(u,v)$ such that $S'(v)$ contains no terminal (since they are useless in a solution to $I$). This way, we obtain a new tree $S''$ rooted at $r$, spanning all terminals, and such that $\ell(S'')<\ell(S)$ and $S''(v) \cap T \neq \emptyset$ $\forall v \in S''$. This is illustrated on Figure \ref{fig:cor77}.

Let $w'$ be a vertex in $W\setminus\{w\}$, and let $t$ be any terminal in $S(w')$. In $S''$, there is a path from $r$ to $t$, while there is no path from $w$ to $t$. So $w''$, which is the closest vertex to $w$ on $\mu'_{rw}$ verifying $t \in S''(w'')$, has outdegree at least 2 in $S''$.

Let $\hat{w}$ be the vertex with outdegree at least 2 which is the closest to $r$ in $S''$ ({\color{black}$\hat{w}$ belongs to $\mu_{rw}' \setminus \{w\}$, from the previous paragraph}), and let $\mu_{r\hat{w}}''$ be the path from $r$ to $\hat{w}$ in $S''$. Note that the only vertices $v$ in $S''$ with $S''(v) \cap T =T$ are those on $\mu_{r\hat{w}}''$, and that all arcs on $\mu_{r\hat{w}}''$ have capacity at least $K$ since they also belong to $\mu_{rw}'$.
Moreover, since $S''(v) \cap T \neq \emptyset$ for every child $v$ of $\hat{w}$ in $S''$, all arcs in $S''(\hat{w})$ only need to have capacity $K-1$. Hence, $S''$ is a feasible solution to $I$ with $\ell(S'') < \ell(S)$, a contradiction.
\end{proof}

\begin{figure}
	\begin{center}
		\includegraphics[scale=0.95]{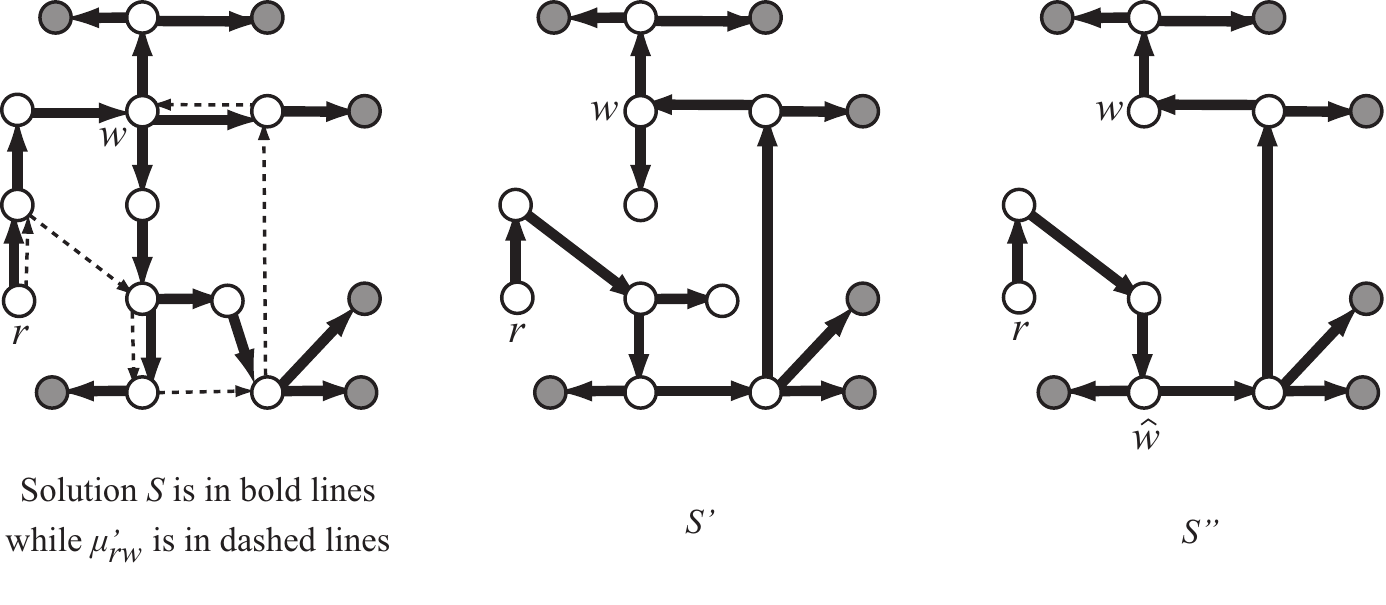}
		\caption{Illustration of the proof of Claim \ref{claim:shortpath}.}
		\label{fig:cor77}
	\end{center}
\end{figure}

We now consider 
 \texttt{ML-CAP-STEINER-TREE} with $c_{\min} \geq K-1$ and show that, when $K$ is fixed, it is solvable in polynomial time.

 \begin{theorem}
 \label{th:large-capas-fixedK}
 If $K$ is fixed and $c_{\min}\geq K-1$, \texttt{ML-CAP-STEINER-TREE} can be solved by an algorithm whose running time is polynomial, and whose only non FPT factor with respect to K is $O(n^{O(\log(K))})$. In particular, \texttt{ML-CAP-STEINER-TREE} is solvable in polynomial time if $K=2$.
 \end{theorem}

  \begin{proof}
  The case  $c_{\min} \geq K$ has already been settled at the beginning of this section. So, assume $c_{\min} = K-1$. Given an instance $I$ of \texttt{ML-CAP-STEINER-TREE}, we solve $I$ as follows.
  We consider all vertices $w$ such that $w$ is either the root $r$ or  a  vertex of degree at least 3 in $G$. For each such vertex $w$:
  \begin{itemize}
  \vspace{-0.2cm}\item We determine a shortest path $\mu_{rw}$ from $r$ to $w$ in $G'=(V,E_K)$, and we denote by $G_w$ the subgraph of $G$ obtained by removing all vertices of $\mu_{rw}$, except $w$;
  \vspace{-0.3cm}\item We consider all pairs $(t_i,t_j)$ of distinct terminals and all subsets $W$ of at most $2\log_2(K)-2$ vertices  $v\neq w$ of degree at least 3 in $G_w$.
  \end{itemize}

  \vspace{-0.2cm}So, let $(w,t_i,t_j,W)$ be such a quadruple. We first determine two internally vertex-disjoint paths $\mu_{wt_i}$ and $\mu_{wt_j}$ linking $w$ to $t_i$ and to $t_j$ {\color{black}in $G_w$}, and such that $\mu_{wt_i}\cup\mu_{wt_j}$ contains all vertices of $W$ and has minimum total length. As in the proof of Theorem \ref{th:CSTP-cap1}, this can be done in polynomial time: we add a sink $s$ and two arcs $(t_i,s)$ and $(t_j,s)$, and we determine two internally vertex-disjoint paths of minimum total length from $w$ to $s$ by using a min-cost flow algorithm; in addition, we impose a flow equal to 1 on each arc $(v',v'')$  corresponding to a vertex $v$ of $W$ in the graph $H$ obtained from $G_w$ (as in the proof of  Theorem \ref{th:CSTP-cap1}), to ensure that the paths contain $W$.

 Assume we are able to find the two internally vertex-disjoint paths $\mu_{wt_i}$ and $\mu_{wt_j}$ in $G_w$. We then consider the graph $G'_w$ obtained from $G_w$ by assigning a length 0 to all arcs on $\mu_{wt_i}\cup\mu_{wt_j}$, and we determine a directed tree $S_{wt_it_j}$ of minimum total length in $G'_w$, rooted at $w$, and spanning all terminals in $T\setminus\{t_i,t_j\}$. Let $R_{wt_it_j}$ denote the set of arcs $(u,v)$ in $S_{wt_it_j}$ with $v$ not belonging to $\mu_{wt_i}\cup\mu_{wt_j}$. We finally build a solution to $I$ by taking all arcs of $\mu_{rw}\cup\mu_{wt_i}\cup\mu_{wt_j}\cup R_{wt_it_j}$.

 Among all built solutions, we keep the best one, which we denote by $S_{best}$. We now prove that $S_{best}$ is an optimal solution to $I$. Let $S^*$ be an optimal solution to $I$, and let $w$ be the vertex in $S^*$ the closest to $r$ with out-degree at least 2. Let $v_1$ and $v_2$ be two children of $w$ in the skeleton of  $S^*$, and  let $t_i$ (resp. $t_j$) be a terminal in $S^*(v_1)$ (resp. $S^*(v_2)$) closest to $v_1$ (resp. $v_2$) in terms of the number of vertices on the path linking them in the skeleton of $S^*(v_1)$ (resp. $S^*(v_2)$).
 We denote by $\mu^*_{rw},\mu^*_{wt_i}$ and $\mu^*_{wt_j}$ the paths in $S^*$ linking $r$ to $w$, $w$ to $t_i$, and $w$ to $t_j$, respectively. Finally, let $W$ be the set of vertices $v\neq w$ on $\mu^*_{wt_i}\cup\mu^*_{wt_j}$ having degree at least 3 in $S^*$, and let $R^*$ denote the set of arcs in $S^*$ that do not belong to $\mu^*_{rw}\cup\mu^*_{wt_i}\cup\mu^*_{wt_j}$.

  We first prove that the proposed algorithm considers the quadruple $(w,t_i,t_j,W)$. Since $w$ has out-degree at least 2, it is either the root $r$ or a vertex of degree at least 3 in $G$. It follows from Claim \ref{claim:shortpath} that $G_w$ contains all vertices of $W$. Hence, we only have to prove that $|W|\leq 2\log_2(K)-2$. Let $n_1$ (resp. $n_2$) be the number of vertices in the skeleton of $S^*(v_1)$ (resp. $S^*(v_2)$). If $n_1=1$, the path from $v_1$ to $t_i$  in the skeleton of $S^*$ is reduced to $1=\log_2(n_1+1)$ vertex $v_1=t_i$; otherwise, $v_1$ has out-degree at least 2 in $S^*$, and we know from {\color{black}the proof of} Property \ref{prop:pathlength} that the path from $v_1$ to $t_i$ in the skeleton of $S^*$ has at most $\log_2(n_1+1)$ vertices. Similarly, the path from $v_2$ to $t_j$ in the skeleton of $S^*$ has at most $\log_2(n_2+1)$ vertices. Hence $W$ contains at most $\log_2(n_1+1)+\log_2({\color{black}n_2}+1)-2$ vertices. Since $w$ has out-degree at least 2, it follows from Property \ref{prop:numbvert} that the skeleton of $S^*(w)$ contains at most $2K-1$ vertices, which implies $n_1+n_2\leq 2K-2$. The sum $\log_2(n_1+1)+\log_2({\color{black}n_2}+1)-2$ is therefore maximized for $n_1=n_2=K-1$, which implies that $W$ contains at most $2\log_2(K)-2$ vertices.

 We now prove that $\ell(S_{best})\leq \ell(S^*)$.
 Let  $\mu_{rw}\cup\mu_{wt_i}\cup\mu_{wt_j}\cup R_{wt_it_j}$ be the solution returned by the proposed algorithm for the quadruple $(w,t_i,t_j,W)$.
It follows from Claim \ref{claim:shortpath} that $\ell(\mu_{rw})=\ell(\mu^*_{rw})$, and that $G_w$ contains all vertices of $S^*(w)$. Since $\mu^*_{wt_i}$ and $\mu^*_{wt_j}$ are two internally vertex-disjoint paths linking $w$ to $t_i$ and to $t_j$ in $G_w$, we have
 $\ell(\mu_{wt_i})+\ell(\mu_{wt_j})\leq\ell(\mu^*_{wt_i})+\ell(\mu^*_{wt_j})$.
 Consider now the set $R'^*$ of arcs $(u,v)$ in $R^*$ with $v$ not belonging to $\mu_{wt_i}\cup\mu_{wt_j}$, and let $S$ be the tree
 obtained from $S^*(w)$ by replacing $\mu^*_{wt_i}$ and $\mu^*_{wt_j}$ by $\mu_{wt_i}$ and $\mu_{wt_j}$, and by {\color{black}removing the arcs in $R^* \setminus R'^*$}. Note that $S$ is a tree rooted at $w$, spanning all terminals in $T\supset T\setminus\{t_i,t_j\}$, and with total length at most equal to $\ell(R^*)$ in $G'_w$ (since all arcs in $\mu_{wt_i}\cup\mu_{wt_j}$ have length 0 in $G'_w$).  Hence, $\ell(R_{wt_it_j})\leq\ell(S_{wt_it_j})\leq \ell (S)\leq \ell(R^*)$ in $G'_w$, which implies $\ell(R_{wt_it_j})\leq\ell(R^*)$ in $G$. In summary,
 \begin{align*} \ell(S_{best})&{ \color{black}= }\ell(\mu_{rw})+\ell(\mu_{wt_i})+\ell(\mu_{wt_j})+\ell(R_{wt_it_j})\\
 	&\leq \ell(\mu^*_{rw})+\ell(\mu^*_{wt_i})+\ell(\mu^*_{wt_j})+\ell(R^*)\\
 	&=\ell(S^*).
 \end{align*}

The total number of possible quadruples is {\color{black}$O(K^2 n^{O(\log(K))})$}.  For each of them, we have to compute a shortest path, a minimum-cost flow, and {\color{black}an optimal Steiner tree spanning $K-2$ terminals. The latter problem can be solved in time FPT with respect to the number of terminals \cite{dreyfus,feldman,watel}, and the other two problems can be solved in polynomial time \cite{AHUJA}.}
 \end{proof}

Together with Theorem \ref{th:undirected-unif-fixedK}, the previous theorem shows, in particular, that Theorems \ref{th:CSTP-directed-unif-K-fixed} and \ref{th:CSTP-undirected-non-unif} are best possible. {\color{black}It also implies, together with Theorem \ref{th:CSTP-cap1}, that \texttt{ML-CAP-STEINER-TREE} is polynomial-time solvable if $K=3$ and all capacities are equal}. When $K$ is part of the input (i.e., not fixed), we have the following result, which complements Theorem \ref{th:unit-cost-Steiner}.

  \begin{theorem}
  \label{th:large-capas-non-fixedK}
   If $c_{min}\geq K-1$, \texttt{CAP-STEINER-TREE} is solvable in polynomial time and \texttt{ML-CAP-STEINER-TREE} can be approximated within a ratio of $1+\rho$.
  \end{theorem}
  \begin{proof}
 We use the same ideas as those used in {\color{black}the proof of} the previous theorem. More precisely, for solving an instance $I$ of \texttt{CAP-STEINER-TREE}, we enumerate all triples $(w,t_i,t_j)$, where $w$ is either the root $r$ or a  vertex of degree at least 3 in $G$, and
 $t_i, t_j$ both belong to $T$. For each such triple, we determine a shortest path $\mu_{rw}$ from $r$ to $w$ in $G'=(V,E_K)$, remove all vertices of $\mu_{rw}$, except $w$, to create $G_w$, and determine two internally vertex-disjoint paths $\mu_{wt_i}$ and $\mu_{wt_j}$ from $w$ to $t_i$ and to $t_j$ in $G_w$. This can be done in polynomial time using path and flow techniques similar to those used in the previous proof (without lengths on the arcs). If we succeed in finding the three paths
 $\mu_{rw},\mu_{wt_i},\mu_{wt_j}$ for a triple $(w,t_i,t_j)$, then we greedily complete their union into a Steiner tree rooted at $r$ and spanning all terminals, which gives a solution to $I$. Otherwise, $I$ does not have any feasible solution. All this can be done in polynomial time, since there are $O(nK^2)$ triples $(w,t_i,t_j)$ to enumerate.

 For an instance $I'$ of \texttt{ML-CAP-STEINER-TREE}, we again enumerate all triples $(w,t_i,t_j)$, and determine for each such triple a shortest path $\mu_{rw}$ from $r$ to $w$ in $G'$, as well as two internally vertex-disjoint paths of shortest total length, $\mu_{wt_i}$ and $\mu_{wt_j}$ from $w$ to $t_i$ and to $t_j$ in $G_w$ (as in the proof of Theorem \ref{th:large-capas-fixedK}).  If we succeed in finding the three paths
 $\mu_{rw},\mu_{wt_i},\mu_{wt_j}$ for a triple $(w,t_i,t_j)$, we then use a $\rho$-approximation algorithm to determine a directed tree $S_{wt_it_j}$ of minimum total length, rooted at $r$, and spanning all terminals in $T\setminus\{t_i,t_j\}$.
 Let $R_{wt_it_j}$ be the set of arcs $(u,v)$ in $S_{wt_it_j}$ with $v$ not belonging to $\mu_{rw}\cup\mu_{wt_i}\cup\mu_{wt_j}$; we build a solution to $I'$ by taking all arcs of $\mu_{rw}\cup\mu_{wt_i}\cup\mu_{wt_j}\cup R_{wt_it_j}$. Among all built solutions, we keep the best one, which we denote by $S_{best}$. Now, let $S^*$ be an optimal solution to $I'$, and let $w$ be the vertex in $S^*$ the closest to $r$ with out-degree at least 2.
 Let $v_1$ and $v_2$ be two distinct children of $w$ in $S^*$, and let $t_i$ be a terminal in $S^*(v_1)$, and $t_j$ a terminal in $S^*(v_2)$. The triple $(w,t_i,t_j)$ is considered in our enumeration, and we clearly have $\ell(\mu_{rw})+\ell(\mu_{wt_i})+\ell(\mu_{wt_j})\leq \ell(S^*)$ and
 $\ell(R_{wt_it_j})\leq \ell(S_{wt_it_j})\leq \rho\ell(S^*)$. Hence, $\ell(S_{best})\leq (1+\rho)\ell(S^*)$.

 Again, the whole process takes a polynomial time. Indeed, there are $O(nK^2)$ enumerated triples, and, for each of them, we have to determine a shortest path, a minimum-cost flow, and a $\rho$-approximate solution to an instance of \texttt{STEINER-TREE}. All these problems can be solved in polynomial time.
  \end{proof}

\section{Concluding remarks}

{\color{black}We have studied the complexity of \texttt{ML-CAP-STEINER-TREE} in digraphs, DAGs and undirected graphs, and we have dealt with any possible case with respect to all the parameters that we considered ({\color{black}minimum and maximum capacities}, lengths, and number of terminals).} {\color{black}Moreover, whenever \texttt{ML-CAP-STEINER-TREE} was intractable while \texttt{CAP-STEINER-TREE}, the case with lengths $0$, was not, we have provided approximation results for \texttt{ML-CAP-STEINER-TREE} nearly as good as the best ones for \texttt{STEINER-TREE}.}

While we have also obtained some results about the parameterized complexity of ML-CAP-STEINER-TREE, several questions remain open in this area:
\begin{itemize}
\vspace{-0.2cm}	\item The results associated with leaves 11 and 13 in Figure \ref{fig:arbres} are best possible, since the FPT-reduction from VDISJ-PATH parameterized by $p$ described in Theorem \ref{cor:pMVDP-unif} shows in particular that CAP-STEINER-TREE is W[1]-hard with respect to either $K$ or $\kappa$ in this case, even with uniform capacities (note that, in this reduction, we have $K = O(p^2)$ and $\kappa = O(p^2)$).
\vspace{-0.2cm}	\item However, the result associated with leaf 9 in Figure 1 may not be the best possible one (i.e., this case might actually be FPT with respect to $\kappa$), since in undirected graphs VDISJ-PATH is FPT with respect to $p$ (so Theorem \ref{cor:pMVDP-unif} does not provide any useful information in this case).
\vspace{-0.2cm}	\item Finally, we think that the main open problem is related to the result provided in Theorem \ref{th:large-capas-fixedK} (and associated to leaf 5 in Figure \ref{fig:arbres}). We have proved that ML-CAP-STEINER-TREE is polynomial-time solvable in this case, hence generalizing the same result already known for STEINER-TREE, but it may actually be FPT with respect to $K$: in particular, notice that this is indeed the case for STEINER-TREE.
\end{itemize}


\section*{Acknowledgments}

This work was done with the support of the Gaspard Monge Program for Optimization and operations research (PGMO) http://www.fondation-hadamard.fr/fr/pgmo.



\begin{thebibliography}{99}

\bibitem{AHUJA}
{\sc R.K. Ahuja}, {\sc T. L.  Magnanti}, {\sc J.B. Orlin},
Networks flows: Theory, Algorithm, and Applications,
{\it Prentice Hall} (1993).

\bibitem{arkin}
E. M. Arkin, N. Guttmann-Beck, R. Hassin (2012), {\it The ($K,k$)-Capacitated Spanning Tree Problem}, Discrete Optimization 9, 258--266.


\bibitem{bern}
M Bern, P Plassmann (1989), {\it The Steiner tree problem with edge lengths 1},  Information Processing Letters 32, 171--176.

\bibitem{bjorklund}
A. Björklund, T. Husfeldt (2014), {\it Shortest Two Disjoint Paths in Polynomial Time}, Proceedings ICALP, LNCS 8572, 211--222.

\bibitem{bousba}
C. Bousba, L.A. Wolsey (1991), {\it Finding minimum cost directed trees with demands and capacities}, Annals of Operations Research 33, 285--303.

\bibitem{byrka}
J. Byrka, F. Grandoni, T. Rothvoß, L. Sanità (2010),
  {\it Approximation algorithms for directed Steiner problems}, Proceedings STOC, 583--592.

\bibitem{charikar}
M. Charikar, C. Chekuri, T.-Y. Cheung, Z. Dai, A. Goel, S. Guha, M. Li (1998), {\it An improved LP-based approximation for Steiner tree}, Proceedings SODA, 192--200.

\bibitem{cheng}
X. Cheng, D.-Z. Du (eds.), Steiner Trees in Industry, Springer (2001).

\bibitem{cong}
J. Cong, A. B. Kahng, K.-S. Leung (1998), {\it Efficient algorithms for the minimum shortest path Steiner arborescence problem with applications to VLSI physical design}. IEEE Trans. on CAD of Integrated Circuits and Systems 17, 24--39.

\bibitem{downey}
R. G. Downey, M. R. Fellows, Parameterized Complexity,
Springer-Verlag (1999).

\bibitem{dreyfus}
S. E. Dreyfus, R. A. Wagner (1971), {\it The Steiner problem in graphs}, Networks 1, 195--207.

\bibitem{du}
D. Du, X. Hu, Steiner Tree Problems In Computer Communication Networks, {\em World Scientific Publishing} (2008).
	
\bibitem{duan}
G. Duan, Y. Yu (2003), {\it Distribution System Optimization by an Algorithm for Capacitated Steiner Tree Problems with Complex flows and Arbitrary Cost Functions}, International Journal of Electrical Power and Energy Systems 25, 515--523.

\bibitem{edmonds}
J. Edmonds, R.M. Karp (1972),  {\it Theoretical improvements in algorithmic efficiency for network flow problems}, J. of the ACM 19, 248--264.

\bibitem{feige}
U. Feige (1996), {\it A threshold of ln $n$ for approximating set-cover}, Proceedings STOC, 314--318.

\bibitem{feldman}
J. Feldman, M. Ruhl (2006), {\it The Directed Steiner Network problem is tractable for a constant number of terminals}, SIAM Journal on Computing 36, 543--561.

\bibitem{fenner}
T. Fenner, O. Lachish, A. Popa (2014), {\it Min-sum 2-paths problems}, Proceedings WAOA 2013, LNCS 8447, 1--11.


\bibitem{fortune}
S. Fortune, J. Hopcroft, J. Wyllie (1980), {\it The directed subgraph homeomorphism problem}, Theoretical Computer Science 10,  111--121.

\bibitem{garey}
{\sc M.R. Garey}, {\sc D.S. Johnson}, Computers and intractability, a guide to the theory of NP-completeness, {\em ed. Freeman, New York} (1979).

\bibitem{gondran}
{\sc M. Gondran}, {\sc M. Minoux}, Graphs and Algorithms, Chapter 5, {\em ed. Wiley} (1984).

\bibitem{haj}
M. Hajiaghayi , R. Khandekar, G. Kortsarz and Z. Nutov (2014), {\it On fixed cost k-flow problems}, Proceedings WAOA 2013, LNCS 8447, 49--60.


\bibitem{hertz}
A. Hertz, O. Marcotte, A. Mdimagh, M.Carreau, F. Welt (2012),
{\it Optimizing the Design of a Wind Farm Collection Network}, INFOR, 95--104.


\bibitem{hwang}
F. K. Hwang, D. S. Richards, P. Winter (1992).{ \it The Steiner Tree Problem}. Annals of Discrete Mathematics 53. North-Holland: Elsevier.

\bibitem{jothi}
R. Jothi, B. Raghavachari (2005). {\it Approximation Algorithms for the Capacitated Minimum Spanning Tree Problem and Its Variants in Network Design}. ACM Transactions on Algorithms 1--2, 265--282.

\bibitem{kobayashi}
Y. Kobayashi, C. Sommer (2010). {\it  On Shortest Disjoint Paths in Planar Graphs}, Discrete Optimization 7, 234--245.

\bibitem{lee}
K. Lee, K. Park, S. Park (1996). {\it Design of capacitated networks with tree configurations}. Telecommunication Systems 6--1, 1--19.

\bibitem{papadimitriou}
C. H. Papadimitriou (1978), {\it The complexity of the capacitated tree problem}, Networks 8, 217--230.

\bibitem{pillai}
A.C. Pillai, J. Chick, L. Johanning, M. Khorasanchi, V. de Laleu (2015), Engineering Optimization, 47--12, 1689--1708.

\bibitem{promel}
H. J. Prömel, A. Steger, The Steiner Tree Problem, Advanced Lectures in Mathematics, {\em ed. Springer} (2002).

\bibitem{robertson}
N. Robertson, P.D. Seymour (1995), {\it  Graphs minors XIII: The disjoint paths problem} J. Comb. Theory, Series B 63,  65--110.

\bibitem{robins}
G. Robins, A. Zelikovsky (2000), {\it Improved Steiner tree approximation in graphs}, Proceedings SODA, 770--779.

\bibitem{schrijver}
{\sc A. Schrijver},
Combinatorial Optimization, Polyhedra and Efficiency,
{\it Springer-Verlag} (2003).

\bibitem{slivkins}
A. Slivkins  (2003) {\it Parameterized Tractability of Edge-Disjoint Paths on Directed Acyclic Graphs}, Proceedings ESA, 482--493.

\bibitem{tovey}
C. A. Tovey (1984)  {\it A simplified NP-complete satisfiability problem}, Discrete Applied Mathematics 8, 85--89.

\bibitem{uchoa}
E. Uchoa, R. Fukasawa, J. Lysgaard, A. Pessoa, M. Poggi de Aragão, D. Andrade (2008), {\it Robust branch-cut-and-price for the Capacitated Minimum Spanning Tree problem over a large extended formulation}, Math. Program., Ser. A 112, 443--472.

\bibitem{voss}
S. Voß, (2001) { \it Capacitated minimum spanning trees}, in Encyclopedia of Optimization, C.A. Floudas and P.M. Pardalos (Editors), Kluwer, Vol. 1, 25--235.

\bibitem{watel}
D. Watel, M.-A. Weisser, C. Bentz, D. Barth (2015) {\it An FPT algorithm in polynomial space for the Directed Steiner Tree problem with Limited number of Diffusing nodes}, Information Processing Letters 115, 275--279.

\bibitem{west}
{\sc D.B. West},
Introduction to Graph Theory,
{\it Second Edition, Prentice Hall} (2001).

\bibitem{wu}
B.Y. Wu (2012), {\it On the maximum disjoint paths problem on edge-colored graphs}, Discrete Optimization 9,  50--57.

\bibitem{zosin}
L. Zosin, S. Khuller (2002), {\it On directed Steiner trees}, Proceedings SODA, 59--63.
\end{thebibliography}
\end{document}